\documentclass[11pt,letterpaper]{article}
\usepackage{graphicx} 
\usepackage[utf8]{inputenc}
\usepackage[margin=1in]{geometry}
\usepackage{amsmath,amsthm,amssymb}
\usepackage{bbm}
\usepackage{xcolor}
\usepackage{hyperref}
\usepackage{float}
\usepackage{subcaption}
\usepackage[nameinlink,capitalise]{cleveref}
\usepackage{natbib}
\usepackage{bbold}
\usepackage{mathtools}
\usepackage{doi}
\usepackage{tikz}
\usepackage{thm-restate}
\usepackage{algorithm, algorithmic}
\usepackage{enumitem}
\usepackage{hyperref}
\usepackage{url}
\usepackage{amsthm, amsfonts, amsmath, amssymb}
\usepackage{booktabs}
\usepackage{graphicx}
\usepackage{enumitem}
\usepackage{pifont}

\usepackage{crossreftools}

\crefname{criterion}{Rule}{Rules}

\definecolor{DarkGreen}{rgb}{0.1,0.5,0.1}
\definecolor{DarkRed}{rgb}{0.5,0.1,0.1}
\definecolor{DarkBlue}{rgb}{0.1,0.1,0.5}
\definecolor{Gray}{rgb}{0.2,0.2,0.2}

\hypersetup{
    unicode=false,          
    pdftoolbar=true,        
    pdfmenubar=true,        
    pdffitwindow=false,      
    pdfnewwindow=true,      
    colorlinks=true,       
    linkcolor=DarkBlue,          
    citecolor=DarkGreen,        
    filecolor=DarkRed,      
    urlcolor=DarkBlue,          
    pdftitle={},
    pdfauthor={},
}

\newcommand{\cmark}{\ding{51}\hfill}
\newcommand{\xmark}{\ding{55}\hfill}

\usepackage{amsmath,amsfonts,bm}

\def\1{\bm{1}}

\DeclareMathAlphabet{\mathsfit}{\encodingdefault}{\sfdefault}{m}{sl}
\SetMathAlphabet{\mathsfit}{bold}{\encodingdefault}{\sfdefault}{bx}{n}

\DeclareMathOperator*{\argmax}{arg\,max}

\newcommand{\PNEBound}{\text{NE}}
\newcommand{\NOPNEBound}{\text{NoNE}}

\newcommand{\uniform}{\text{Uniform}}

\newcommand{\bfmu}{\boldsymbol{\mu}}
\newcommand{\bfX}{\boldsymbol{X}}
\newcommand{\bfW}{\boldsymbol{W}}

\newcommand{\bfpi}{\boldsymbol{\pi}}

\newcommand{\calM}{\mathcal{M}}

\newcommand{\calC}{\mathcal{C}}

\newcommand{\calD}{\mathcal{D}}

\newcommand{\calG}{\mathcal{G}}
\newcommand{\calS}{\mathcal{S}}

\newcommand{\calJ}{\mathcal{J}}

\newcommand{\calU}{\mathcal{U}}
\newcommand{\calZ}{\mathcal{Z}}
\newcommand{\calK}{\mathcal{K}}

\newcommand{\bbE}{\mathbb{E}}
\newcommand{\bbP}{\mathbb{P}}
\newcommand{\bbI}{\mathbb{I}}

\newcommand{\sort}{\mathrm{Sort}}
\newcommand{\discontent}{D}
\newcommand{\watchful}{C^-}
\newcommand{\hopeful}{C^+}
\newcommand{\content}{C}
\newcommand{\learn}{\text{learn}}
\newcommand{\init}{\text{init}}

\DeclareMathOperator{\regret}{Reg}

\DeclareMathOperator{\noneq}{NonPNE}

\theoremstyle{plain}
\newtheorem{theorem}{Theorem}[section]

\newtheorem{lemma}[theorem]{Lemma}

\theoremstyle{definition}
\newtheorem{definition}{Definition}[section]

\newtheorem{example}{Example}[section]
\newtheorem{criterion}{Rule}[section]
\newtheorem{observation}{Observation}[section]

\theoremstyle{remark}
\newtheorem{remark}{Remark}[section]

\newenvironment{enum}
{\begin{enumerate}[noitemsep,topsep=0pt,parsep=0pt,partopsep=0pt]}
{\end{enumerate}}

\makeatletter
\def\maketag@@@#1{\hbox{\m@th\normalfont\normalsize#1}}
\makeatother

\title{PPA-Game: Characterizing and Learning Competitive Dynamics Among Online Content Creators}

\newcommand*\samethanks[1][\value{footnote}]{\footnotemark[#1]}

\author{Renzhe Xu\thanks{MoE Key Laboratory of Interdisciplinary Research of Computation and Economics, Institute for Theoretical Computer Science, Shanghai University of Finance and Economics, Shanghai, China. Email: \texttt{xurenzhe@sufe.edu.cn}.}
\and
Haotian Wang\thanks{College of Computer Science and Technology, National University of Defense Technology, Changsha, China. Email: \texttt{accwht@hotmail.com}.}
\and
Xingxuan Zhang\thanks{Department of Computer Science and Technology, Tsinghua University, Beijing, China. Emails: \texttt{xingxuanzhang@hotmail.com}, \texttt{cuip@tsinghua.edu.cn}.}
\and
Bo Li\thanks{School of Economics and Management, Tsinghua University, Beijing, China. Email: \texttt{libo@sem.tsinghua.edu.cn}.}
\and
Peng Cui\samethanks[3]
}

\date{}

\begin{document}

\maketitle

\begin{abstract}
    In this paper, we present the Proportional Payoff Allocation Game (PPA-Game), which characterizes situations where agents compete for divisible resources. In the PPA-game, agents select from available resources, and their payoffs are proportionately determined based on heterogeneous weights attributed to them. Such dynamics simulate content creators on online recommender systems like YouTube and TikTok, who compete for finite consumer attention, with content exposure reliant on inherent and distinct quality. We first conduct a game-theoretical analysis of the PPA-Game. While the PPA-Game does not always guarantee the existence of a pure Nash equilibrium (PNE), we identify prevalent scenarios ensuring its existence. Simulated experiments further prove that the cases where PNE does not exist rarely happen. Beyond analyzing static payoffs, we further discuss the agents' online learning about resource payoffs by integrating a multi-player multi-armed bandit framework. We propose an online algorithm facilitating each agent's maximization of cumulative payoffs over $T$ rounds. Theoretically, we establish that the regret of any agent is bounded by $O(\log^{1 + \eta} T)$ for any $\eta > 0$. Empirical results further validate the effectiveness of our online learning approach.
\end{abstract}
\section{Introduction}
Online recommender systems, exemplified by platforms such as YouTube and TikTok, are now central to our daily digital experiences. From an economic viewpoint, these systems represent intricate ecosystems composed of multiple stakeholders: consumers, content creators, and the platform itself ~\citep{boutilier2023modeling,ben2020content}. Although a significant portion of current research delves into consumer preferences and individualized recommendations, the role and behavior of content creators are also crucial ~\citep{zhan2021towards,boutilier2023modeling}. They not only shape the pool of recommended content but also influence the system's stability, fairness, and long-term utility ~\citep{ben2018game,patro2020fairrec,boutilier2023modeling,mladenov2020optimizing}. A major challenge when capturing the dynamics of content creators is characterizing the competitive exposure among them, given the finite demand for different content. Consequently, understanding the competitive dynamics among content creators becomes paramount.

Game theory offers a useful framework to characterize competitive behaviors among multiple players ~\citep{nisan2007algorithmic}. While numerous studies have employed game-theoretical analysis on content creators within online platforms, many ~\citep{yao2024user,yao2024unveiling,yao2023rethinking,ben2018game,hron2022modeling,jagadeesan2022supply,ben2019recommendation} primarily explore Nash equilibrium properties, largely overlooking the dynamics of content creators and their convergence to the equilibrium. Some research attempts to bridge this gap, considering different allocation models for creators producing analogous content ~\citep{ben2019convergence,ben2020content,xu2023competing,liu2020competing,liu2021bandit}. For instance, \citet{ben2019convergence,ben2020content,liu2020competing,liu2021bandit} assumed that only top-tier content creators receive exposure, and \citet{xu2023competing} proposed an averaged allocation rule. While these dynamics offer valuable insights, they also have limitations due to the gap between their modeling assumptions and real-world contexts. For instance, top-winner dynamics~\citep{ben2019convergence,ben2020content,liu2020competing,liu2021bandit} may conflict with diversity principles in online recommendation platforms, while an averaged allocation rule~\citep{xu2023competing} could overlook the nuanced specializations that creators bring to diverse content domains.

To bridge the gap between theoretical models and the complex dynamics observed among online content creators, we introduce the Proportional Payoff Allocation Game (PPA-Game) as a novel model to capture the competitive dynamics of online content creators. Adapting from ~\citep{ben2019convergence,ben2020content,prasad2023content}, we assume a game with $N$ players (content creators) and $K$ resources (content topics). In this context, players need to decide which resources to select. Each resource has an associated total payoff (indicative of the exposure for each content topic), and individual players have specific weights (or quality metrics) assigned to each resource. When several players choose the same resource, its payoff is distributed proportionally according to player weights—a mechanism resonating with other game-theoretical frameworks ~\citep{caragiannis2016welfare,kelly1997charging,hoefer2017proportional,anbarci2023proportional}.

We first explore the existence of the pure Nash equilibria (PNE) within the game, revealing that PNE might not exist in particular scenarios. We further identify several prevalent conditions where PNE's existence is guaranteed. These conditions, with broad applicability in real-world scenarios, guarantee PNE existence if any of the following holds: (1) The payoffs of the resources exhibit a long-tailed distribution, meaning the ratio between the payoffs of any two resources either exceeds a specific constant $c_0 > 1$ or is less than its reciprocal $1/c_0$.
(2) All players maintain the same weights across different resources. (3) Every resource has identical payoffs. To further investigate the probability of PNE's existence, we conducted simulation experiments. Our results indicate that such situations where PNE is absent rarely occur.


Leveraging the PPA-Game, we introduce a Multi-player Multi-Armed Bandit (MPMAB) framework ~\citep{rosenski2016multi,besson2018multi,boursier2019sic,boursier2020selfish,shi2020decentralized,shi2021heterogeneous,huang2022towards,lugosi2022multiplayer} to simulate environments where content creators lack prior information about content demand and their quality in various topics. In this context, the `players' represent content creators, while `arms' denote content topics. This model entails players making simultaneous decisions over multiple rounds, with payoffs determined by the PPA-Game mechanism. Inspired by \citet{xu2023competing,boursier2020selfish}, we assume that selfish players, motivated by their own interests, aim to maximize their own cumulative payoffs over $T$ rounds. Within this MPMAB structure, inspired by ~\citep{bistritz2018distributed,bistritz2021game,pradelski2012learning}, we delineate a novel online learning strategy for players and prove that each player's regret is bounded by $O(\log^{1 + \eta} T)$ for any $\eta > 0$. Synthetic experiments validate the effectiveness of our approach.

In summary, our key contributions include:

\begin{itemize}
    \item We introduce the novel Proportional Payoff Allocation Game (PPA-Game) to characterize competitive dynamics among online content creators.
    \item We comprehensively examine the PNE's existence within the PPA-Game, identifying broad scenarios where the existence of PNE is prevalent, as proved by simulation results.
    \item Building on the PPA-Game, we introduce a new Multi-player Multi-Armed Bandit (MPMAB) paradigm for simulating real-time decision-making by content creators.
    \item We propose a novel online algorithm, with bounded regret, for modeling player competition. Empirical results demonstrate its efficacy.
\end{itemize}

The remainder of this paper is structured as follows. \cref{sect:related-works} provides a systematic review of related work. \cref{sect:game-formulation} introduces the PPA-Game, examines its theoretical properties, and characterizes the conditions for PNE existence. In \cref{sect:online}, we present the MPMAB framework, propose an online learning algorithm, and establish its theoretical guarantees. \cref{sect:experiments} reports synthetic experiment results, and \cref{sect:conclusion} summarizes our contributions.
\section{Related Works} \label{sect:related-works}
\paragraph{Game-theoretical Analysis for Content Creators} While a significant portion of existing research concentrates on consumer preferences and individualized recommendations \citep{boutilier2023modeling,zhan2021towards}, a few delve into the competitive dynamics among content creators, framing this competition through game-theoretical models \citep{yao2024user,yao2024unveiling,boutilier2023modeling,ben2018game,hron2022modeling,jagadeesan2022supply,ben2019recommendation,ben2019convergence,ben2020content,xu2023competing,liu2020competing,liu2021bandit,yao2023rethinking}. For instance, works by \citet{hron2022modeling,jagadeesan2022supply} model content creator strategies as vectors within a designated space, offering insights into various Nash equilibrium properties. Meanwhile, \citet{ben2019recommendation} conceptualized this competition using facility location games \citep{anderson1992discrete}, though they did not incorporate the learning dynamics inherent to content creators.

A subset of these studies attempts to capture the dynamics among content creators \citep{ben2019convergence,ben2020content,liu2020competing,liu2021bandit,xu2023competing}. Specifically, studies like \citet{ben2019convergence,ben2020content,liu2020competing,liu2021bandit} assume that only top-quality content creators receive exposure for each content topic. This perspective seems overly simplistic, especially considering the platforms' need to equitably distribute exposure to retain a diverse array of creators. On the other hand, \citet{xu2023competing} adopted an averaging allocation rule, suggesting that creators share content topics' exposure equally on average. However, this fails to account for the real-world scenarios where creators bring distinct expertise and excel in producing varied content types.

\paragraph{Proportional Allocation}
Proportional allocation serves as a frequently adopted mechanism in game theory \citep{kelly1997charging,caragiannis2016welfare,syrgkanis2013composable,johari2007price,anbarci2023proportional,hoefer2017proportional}. However, the game formulations in these works vary considerably from our formulation. Works focusing on the proportional allocation mechanism \citep{kelly1997charging,caragiannis2016welfare,syrgkanis2013composable,johari2007price} delves into the distribution of a singular, infinitely divisible resource amongst multiple competing players. In these models, players strategically choose the resource bid values. In contrast, studies like \citet{anbarci2023proportional, klumpp2019dynamics} navigate the Blotto games \citep{arad2012multi,friedman1958game}, where players allocate budgets across various battlefields or contests. Here, the objective hinges on amplifying total winnings, with victory rates in battles proportional to a player's allocated budget. Distinct from these models, our work focuses on multiple divisible resources and players need to choose among different resources.

Furthermore, \citet{hoefer2017proportional} explored a proportional coalition formation game, where players create coalitions (\textit{i.e.}, player groups). Within this framework, coalition payoffs get proportionally divided based on player weights. By contrast, our approach considers the distribution of several divisible resources where identical player groups may gain varying payoffs from different resources.

\paragraph{Multi-player Multi-Armed Bandits}
The Multi-Armed Bandit (MAB) problem~\citep{lai1985asymptotically,bubeck2012regret,slivkins2019introduction,lattimore2020bandit} has widely studied across various settings~\citep{kveton2015cascading,lattimore2016causal,agrawal2016linear,xu2022product,ferreira2022learning}. An extension to the classic MAB is the multi-player multi-armed bandit (MPMAB)~\citep{liu2008restless,jouini2009multi,jouini2010upper,branzei2021multiplayer}, which studies scenarios wherein multiple players interact within a single MAB instance~\citep{boursier2022survey}. Notably, a majority of MPMAB research posits that players receive no payoff in the event of a collision~\citep{rosenski2016multi,besson2018multi,boursier2019sic,boursier2020selfish,shi2020decentralized,shi2021heterogeneous,huang2022towards,lugosi2022multiplayer}, as comprehensively surveyed in~\citep{boursier2022survey}.

Several studies have proposed diverse payoff allocation mechanisms within the MPMAB framework when collisions occur. For instance, works by \citet{liu2020competing,liu2021bandit,jagadeesan2021learning,basu2021beyond} assume that the player with the highest expected payoff get the entire payoff in such scenarios. Others have introduced concepts like thresholds~\citep{bande2019multi,youssef2021resource,bande2021dynamic,magesh2021decentralized} or capacities~\citep{wang2022multi,wang2022multiple} for each arm. In addition, \citet{shi2021multi,pacchiano2021instance} assumed reduced payoffs for each player during collisions. \citet{boyarski2021distributed} presented a model where payoffs are heterogeneous, yet they do not ensure that the sum of payoffs for collided players aligns with the arm's total payoff. Crucially, while many of these studies emphasize solutions that maximize total welfare, we make the assumption of player selfishness, focusing on individual payoff maximization. Lastly, although \citet{xu2023competing} also operated within the MPMAB framework, they advocate an averaging allocation mechanism, a departure from real-world applicability and distinct from our formulation.
\section{Proportional Payoff Allocation Game} \label{sect:game-formulation}
\subsection{Background on Game Theory}
\paragraph{Nash Equilibrium.}
The concepts of $\epsilon$-Nash equilibrium and pure Nash equilibrium (PNE) have been well-established in the game theory literature~\citep{nisan2007algorithmic}. To set the context for our study, consider a game with $N$ players. Let $\calS$ represent the set of potential strategies for these players. Each player $j$ adopts a strategy $\pi_j \in \calS$. The collective strategy of all players is represented as $\bfpi = \{\pi_j\}_{j=1}^N$, while $(\pi', \bfpi_{-j})$ denotes a scenario where player $j$ deviates from his original strategy $\pi_j$ to another strategy $\pi' \in \calS$. We denote the payoff of player $j$ for a given strategy profile $\bfpi$ as $\calU_j(\bfpi)$. With these notations, the $\epsilon$-Nash equilibrium is defined as:

\begin{definition}[$\epsilon$-Nash Equilibrium and Pure Nash Equilibrium (PNE)]
A strategy profile $\bfpi \in \calS^N$ for a game defined by $\calG = \left(\calS, \{\calU_j\}_{j=1}^N\right)$ is termed an $\epsilon$-Nash equilibrium if, for every strategy $\pi' \in \calS$ and for each player $j \in [N]$, $\calU_j(\pi', \bfpi_{-j}) \le \calU_j(\bfpi) + \epsilon$. If $\bfpi$ achieves a $0$-Nash equilibrium, it's termed a pure Nash equilibrium.
\end{definition}

We focus on PNE because they represent the most robust form of equilibrium in game theory~\citep{nisan2007algorithmic} and are essential for ensuring platform stability~\citep{ben2018game}. Previous studies in recommender systems have also emphasized the importance of PNE in understanding system behavior and fairness~\citep{ben2020content, xu2023competing, hron2022modeling}. Furthermore, in the context of our proposed game in \cref{sect:game-formulation-detail}, we note that since it is finite, Mixed Nash Equilibria (MNE) always exist~\citep{nisan2007algorithmic}. However, focusing on PNE is crucial for capturing the most stringent, predictable, and enforceable outcomes in practice.

\paragraph{Total Welfare.}
The total welfare, represented by $W_{\calG}(\bfpi)$, calculates the aggregate payoff derived from a strategy profile $\bfpi$. In the context of our game, it's the summation of the utilities obtained by all participating players, expressed mathematically as: $W_{\calG}(\bfpi) = \sum_{j=1}^N \calU_j(\bfpi)$.

\paragraph{Improvement Step and Improvement Path.}
We leverage the notions of improvement steps and paths from prior works~\citep{ben2019convergence, ben2020content}. Starting with any strategy profile $\bfpi \in \calS^N$, an \textit{improvement step} is a transition to a new strategy profile $(\pi_j', \bfpi_{-j})$, where player $j$ unilaterally deviates to gain a higher payoff, such that $\calU_j(\pi_j', \bfpi_{-j}) > \calU_j(\bfpi)$. Furthermore, an \textit{improvement path} $\ell = \left(\bfpi(1), \bfpi(2), \dots \right)$ is a sequence formed by chaining such improvement steps. Given our focus on a finite game with bounded player counts and strategy space (as shown in \cref{defn:game}), any unending improvement path will inevitably contain cycles, implying the existence of states $t_1 < t_2$ where $\bfpi(t_1) = \bfpi(t_2)$. Crucially, if all potential improvement paths in a game are finite, every better-response dynamics will converge to a PNE~\citep{nisan2007algorithmic}.

\subsection{Game Formulation} \label{sect:game-formulation-detail}
Consider a scenario with $K$ distinct resources labeled by indices $[K] = \{1, 2, \dots, K\}$ and $N$ players, where $N, K \ge 2$. Each resource, denoted by $k$, possesses a total payoff $\mu_k$ in the range $0 < \mu_k \le 1$. Simultaneously, each player, represented as $j$, has a weight ranging between $0 \le w_{j,k} \le 1$ for resource $k$. Every player can choose from a strategy set $[K]$, resulting in an overall strategy profile $\bfpi \in [K]^N$. Crucially, if multiple players opt for the same resource, the total payoff from that resource is distributed proportionally based on the individual player weights. This can be formalized as:

\begin{definition} [Proportional Payoff Allocation Game (PPA-Game)] \label{defn:game}
    Given $K$ resources with associated payoffs $\bfmu = \{\mu_k\}_{k=1}^K$ ($0 < \mu_k \le 1$) and $N$ players characterized by weights $\bfW = \{w_{j,k}: j\in [N], k \in [K]\}$ ($0 \le w_{j,k} \le 1$), the Proportional Payoff Allocation Game (PPA-Game) is formulated as $\calG = \left(\calS, \left\{\calU_j(\bfpi)\right\}_{j=1}^N\right)$, where $\calS = [K]$ and
    \begin{equation} \label{eq:game-formulation}
        \calU_j(\bfpi) = \begin{cases}
            \frac{\mu_k w_{j,k}}{\sum_{j'=1}^N \bbI[\pi_{j'}=k]w_{j',k}} & \text{if } \sum_{j'=1}^N \bbI[\pi_{j'}=k]w_{j',k} > 0, \\
            0 & \text{otherwise}.
        \end{cases} \\
    \end{equation}
    Here $k = \pi_j$ and $\bbI[\cdot]$ is the indicator function.
\end{definition}

It's worth noting that any resource $k$ without players having positive weights will remain unchosen. Consequently, such a resource can be disregarded. Based on this observation, we can assert that for each $k \in [K]$, there exists a $j \in [N]$ such that $w_{j,k} > 0$, and the maximum weight across players is $1$, \textit{i.e.}, $\max_{j \in [N]} w_{j,k} = 1$ for every $k \in [K]$. We then demonstrate that how the PPA-Game can be used to model competitive behaviors among content creators within online recommender systems:

\begin{example}
    In modern recommendation platforms like YouTube and TikTok, content creators compete for viewer attention by producing content on a wide range of topics. In this setting, the content creators are the `players,' while the topics they generate act as the `resources.' The `payoff' associated with each resource reflects the overall attention or exposure a topic receives. Since creators possess varying levels of expertise across different topics, exposure is typically distributed in proportion to their influence or proficiency in each specific topic.

    Several models are closely related to our formulation, particularly those by~\citet{bhawalkar2014weighted, ben2019convergence, ben2020content, liu2020competing, liu2021bandit, xu2023competing}. However, the models in~\citep{ben2019convergence, ben2020content, liu2020competing, liu2021bandit} assume that only the player $j$ with the highest weight $w_{j,k}$ for a resource $k$ receives the payoff. In contrast, \citet{xu2023competing} assume an average payoff allocation for players choosing the same resource. The weighted congestion game model~\citep{bhawalkar2014weighted} assumes that players have identical weights across all resources. Therefore, we argue that these formulations diverge from the more realistic dynamics of online recommender systems, as outlined above.
\end{example}

\subsection{Analysis on Pure Nash Equilibrium (PNE)}

In this subsection, we investigate properties of Pure Nash Equilibria.


\subsubsection{High Probability of PNE Existence}

We first acknowledge that PNE may not always exist in certain scenarios~(see \cref{sect:app-non-existence} for a detailed illustrative example):

\begin{observation}\label{ob:non-existence}
PNE does not exist in specific scenarios.
\end{observation}

While this result may seem discouraging, large-scale simulations suggest that PNEs typically exist in most cases. Our synthetic experiments support this finding, as detailed in \cref{sect:app-high-probability}.

\begin{observation}\label{ob:high-probability}
We randomly select $N$ and $K$ from the range $[10, 100]$ and sample resource payoffs and player weights from four distributions supported on $[0,1]$. Among $20,000,000$ randomly generated game configurations, a PNE fails to exist in only one instance.
\end{observation}

These synthetic experiments allow us to infer that the existence of PNE is a common occurrence in practical scenarios.

\subsubsection{Existence of PNE in Specific Scenarios}
While judging the existence of PNE in general cases is impossible, we prove that PNE exists in several typical scenarios, and each corresponds to some realistic case:

\begin{theorem} [Long-tailed Resource Scenario] \label{thrm:pne-exists}
    Define $N_0$ and $\epsilon_0$ as follows.
    $$
        \begin{aligned}
            N_0 & = \max_{k \in [K]} \sum_{j=1}^N \bbI[w_{j,k} > 0], \\
            \epsilon_0 & = \min \left\{w_{j,k}: j \in [N], k \in [K], w_{j,k} > 0\right\}.
        \end{aligned}
    $$
    Then the PNE exists if
    \begin{equation} \label{eq:thrm-pne-exists}
        \forall k \ne k', \quad \frac{\mu_k}{\mu_{k'}} > \frac{N_0}{\epsilon_0} \quad \text{or} \quad \frac{\mu_k'}{\mu_{k}} > \frac{N_0}{\epsilon_0}.
    \end{equation}
\end{theorem}

\begin{remark}
    \cref{thrm:pne-exists} demonstrates scenarios with a long-tailed distribution of resource payoffs. In particular, \cref{eq:thrm-pne-exists} specifies that the peak payoff among resources substantially surpasses that of other resources. For instance, entertainment-related content on TikTok garners substantially more views than other topics \citep{statista2023}.
\end{remark}

Several degenerate cases also guarantee the existence of a PNE.

\begin{theorem} \label{thrm:pne-exists-special}
    The existence of PNE is guaranteed if any of the following conditions are met:
    \begin{enum}
        \item \textbf{(Partially Heterogeneous Scenario)} Players have the same weights for different resources, \textit{i.e.}, $\forall k, k' \in [K], j \in [N]$, $w_{j,k} = w_{j,k'}$.
        \item \textbf{(Homogeneous Player Scenario)} Players have the same weights within each resource, \textit{i.e.}, $\forall k \in [K]$ and $j, j' \in [N]$, $w_{j,k} = w_{j',k}$.
        \item \textbf{(Homogeneous Resource Scenario)} Resources have the same payoffs, \textit{i.e.}, $\forall k, k' \in [K]$, $\mu_k = \mu_{k'}$.
    \end{enum} 
\end{theorem}
\begin{remark}
    \cref{thrm:pne-exists-special} establishes the existence of PNE in several degenerate cases of the PPA-Game (\cref{defn:game}). Specifically, the partially heterogeneous scenario corresponds to cases where player weights reflect their inherent influence. In online recommender platforms, for example, certain creators—termed "influencers"—exert greater influence across all topics due to their expertise, stature, or authority. Notably, the partially heterogeneous scenario forms a special case of the weighted congestion game~\citep{bhawalkar2014weighted}, where PNE may not always exist, yet we prove its existence in this specific setting.
    The homogeneous player scenario considers cases where all players exert equal influence on every topic, aligning with the model studied in~\citep{xu2023competing}. Meanwhile, the homogeneous resource scenario examines settings where all resources provide equivalent payoffs, as seen in competition for processing resources in modern data centers, where each resource has the same processing power~\citep{benblidia2021renewable}.
\end{remark}

\subsubsection{Non-uniqueness and Inefficiency of PNEs} \label{sect:pne-inefficient}
We further show that PNEs might not be unique. Notably, among multiple PNEs, some can yield a reduced total welfare compared to others:

\begin{example} [Existence of multiple and inefficient PNEs]
    Assume the resources' payoffs to be $[0.6, 0.4, 0.2]$. Suppose each player has the same weights on the resources (as per Condition (2) in \cref{thrm:pne-exists}), with weights $[1, 3/8, 1/4]$, respectively. Under such a setup, two distinct PNEs emerge: $\bfpi^{(1)} = (1, 2, 3)$ and $\bfpi^{(2)} = (2, 1, 1)$. Notably, since $W_{\calG}\left(\bfpi^{(1)}\right) = 1.2$ exceeds $W_{\calG}\left(\bfpi^{(2)}\right) = 1.0$, the latter, $\bfpi^{(2)}$, is a less efficient PNE.
\end{example}

Given the existence of such inefficient PNEs, it becomes preferable for players to converge to an optimal PNE, thereby maximizing the overall welfare.

\subsubsection{Price of Anarchy}
We further analyze the price of anarchy (PoA) in the PPA-Game. When a PNE exists, PoA is defined as the ratio of the maximal social welfare across all strategy profiles to the minimal social welfare among all pure Nash equilibria, i.e.,
\begin{equation} \label{eq:poa}
\mathrm{PoA} = \frac{\max_{\bfpi \in [K]^N}W_{\calG}(\bfpi)}{\min_{\bfpi \in [K]^N: \bfpi \text{ is a PNE}}W_{\calG}(\bfpi)}.
\end{equation}
\begin{theorem} \label{thrm:poa}
    If a PNE exists, then
    $$
    \mathrm{PoA} \le 1 + \frac{\min\{N, K\} - 1}{N}.
    $$
\end{theorem}
\begin{remark}
    Notably, this bound on PoA remains identical to that in \citep{xu2023competing}, which analyzes the homogeneous player scenario in \cref{thrm:pne-exists-special}.
\end{remark}

\section{Online Learning} \label{sect:online}
In real-world scenarios, the payoffs associated with each resource are not static but fluctuate over time, influenced by varying factors such as temporal demand for different content topics. Furthermore, players typically lack prior information about both the payoffs of resources and their own respective weights. To capture the learning dynamics of players in the PPA-Game, we introduce a novel framework based on multi-player multi-armed bandits. This setup allows us to model how players adapt and make decisions based on their historical observations.

\subsection{Multi-Player Multi-Armed Bandits Framework}
Inspired by previous work \citep{xu2023competing, boursier2022survey, liu2020competing}, we adopt a decentralized multi-player multi-armed bandit (MPMAB) framework where each `arm' represents a resource. In this setting, players lack prior information regarding both the rewards offered by arms and their personal weights toward these arms\footnote{In the MPMAB framework, the terms `payoff' and `reward' are synonymous, as are `resource' and `arm'.}. Consequently, players must learn this information and maximize their payoffs through ongoing interactions with the platform.

Suppose there are $T$ rounds. At each round $t \in [T]$, each arm (resource) $k$ offers a stochastic reward $X_k(t)$, independently and identically distributed according to a distribution $\xi_k \in \Xi$, with expected value $\mu_k$, supported on $(0, 1]$. We denote the set of all possible such distributions as $\Xi$. The vector of rewards from all arms at round $t$ is represented as $\bfX(t) = \{X_k(t)\}_{k=1}^K$. During each round $t \in [T]$, $N$ players select arms concurrently, based solely on their past observations. We use $\pi_j(t)$ to indicate the arm chosen by player $j$ and $\bfpi(t) = \{\pi_j(t)\}_{j=1}^N$ to denote the strategy profile of all players at that round.

When multiple players select the same arm, the reward is allocated according to the proportional payoff rule specified in \cref{eq:game-formulation}. Specifically, let $R_j(t)$ represent the reward received by player $j$ at round $t$, given by:
\begin{equation*} 
    R_j(t) = X_{k}(t) \cdot \frac{w_{j, k}}{\sum_{j'=1}^N \bbI[\pi_{j'}(t) = k]w_{j', k}}, \, \text{where} \quad k = \pi_j(t).
\end{equation*}

Our framework aligns with the statistical sensing setting commonly found in MPMAB literature \citep{boursier2022survey}. Specifically, each player can observe both their individual reward $R_j(t)$ and the total reward of the arm they have chosen, denoted as $X_{\pi_j(t)}(t)$.

In addition, consistent with previous works \citep{boursier2020selfish, xu2023competing}, our framework operates under the assumption of selfish players. In this setting, each player is focused on maximizing their own reward rather than collaborating to achieve a globally optimal outcome.

\subsection{Evaluation metrics} \label{sect:metrics}
We consider the following evaluation metrics.
\paragraph{Players' Regret.}
Since we consider the selfish player setting, we evaluate the regret of individual players by contrasting the expected reward garnered through their chosen strategy, $\bfpi(t)$, with the best possible choice for each round $t$. The formal expression for player $j$'s regret is as follows:
\begin{equation}  \label{eq:regret}
    \regret_j(T) = \sum_{t=1}^T\left(\max_{k \in [K]}\calU_j((k, \bfpi_{-j}(t)) - \calU_j(\bfpi(t))\right).
\end{equation}
Notably, $\regret_j(T)$ is a random variable and the randomness comes from the actions $\{\bfpi(t)\}$, which rely on historical random observations $\{X_k(t)\}$.

\paragraph{Number of Non-inefficient PNE Rounds.}
As demonstrated in \cref{sect:pne-inefficient}, the game can have multiple Pure Nash Equilibria (PNE), some of which may be suboptimal in terms of total welfare. Therefore, an important metric we consider is not merely the convergence to any PNE, but specifically to the most efficient PNE. We evaluate this using the number of rounds until round $T$ in which players do not converge to the most efficient PNE. Mathematically, this is expressed as:
\begin{equation} \label{eq:non-eq}
    \noneq(T) = \sum_{t=1}^T \bbI[\bfpi(t) \text{ is not a most efficient PNE of } \calG].
\end{equation}
The metric $\noneq(T)$ provides a measure of how frequently the players end up in a suboptimal equilibrium.

\subsection{Algorithm}
Inspired by \citep{bistritz2018distributed, bistritz2021game, pradelski2012learning}, we introduce a novel online learning algorithm designed to enable players to both learn the game's parameters and optimize their individual rewards concurrently. 
The algorithm's framework is shown in \cref{fig:algorithm} and the pseudo-code is outlined in \cref{alg:all}.

\begin{algorithm}[th]
    \caption{PPA-Game Online Learner}
    \label{alg:all}
    \begin{algorithmic}[1]
        \STATE \textbf{Input:} Player rank $j$, Total number of rounds $T$, Phase length parameters $c_1, c_2, c_3, \eta$, Perturbation range $\Gamma$, Exploration probability $\epsilon$
        \STATE \underline{Exploration phase:}
        \STATE Set $S_{j,k} = 0$ for all $k \in [K]$
        \FOR{$t \leftarrow 1$ \TO $c_1K$}
            \STATE Pull arm $\pi_j(t) \leftarrow (t + j) \mod K + 1$ and set $k \leftarrow \pi_j(t)$
            \STATE Get the arm's total reward $X_k(t)$
            \STATE Update the sum of rewards $S_{j,k} \leftarrow S_{j,k} + X_k(t)$
        \ENDFOR
        \STATE Estimate the rewards of each arm $\hat{\mu}_{j,k} \leftarrow S_{j,k} / c_1$ for $k \in [K]$
        \STATE Randomly sample perturbations $\gamma_{j,k} \sim \uniform(0, \Gamma)$ for each $k \in [K]$

        \STATE \underline{Learning and exploitation phase:}
        \STATE $s \leftarrow 0$
        \WHILE{$t \le T$}
            \STATE \underline{Learning PNE subphase:}
            \STATE $s \leftarrow s + 1$.
            \STATE Initialize the counter $Q_{j,k}^{(s)} \leftarrow 0$.
            \STATE Initialize the status $(m_j, a_j, u_j) \leftarrow (\discontent, 1, 0)$
            \FOR{the next $c_2s^{\eta}$ rounds}
                \STATE $t \leftarrow t + 1$
                \STATE Pull arm $\pi_j(t)$ according to \cref{rule:pull-arm} and set $k \leftarrow \pi_j(t)$
                \STATE Get the arm's total reward $X_k(t)$ and player $j$'s own reward $R_j(t)$
                \STATE Set $u_j' \leftarrow R_j(t)\hat{\mu}_{j,k} / X_k(t) + \gamma_{j,k}$
                \STATE Update the status $(m_j, a_j, u_j)$ according to \cref{rule:transit}
                \IF{$m_j = \content$}
                    \STATE $Q_{j,k}^{(s)} \leftarrow Q_{j,k}^{(s)} + 1$
                \ENDIF
            \ENDFOR
            \STATE \underline{Exploitation subphase:}
            \STATE Find $k \leftarrow \argmax_{k'} \sum_{s'=1}^sQ_{j, k'}^{s'}$
            \FOR{the next $c_32^s$ rounds}
                \STATE $t \leftarrow t + 1$
                \STATE Pull arm $\pi_j(t) \leftarrow k'$
            \ENDFOR
        \ENDWHILE
    \end{algorithmic}
\end{algorithm}

Generally speaking, the algorithm operates in two distinct phases: the Exploration Phase and the Learning and Exploitation Phase.  During the Exploration Phase, each player pulls each arm multiple times. This serves to estimate the arm rewards and facilitates the computation of the most efficient Pure Nash Equilibrium (PNE) for the subsequent phase. Furthermore, in the Learning and Exploitation Phase, players alternate between two subphases. The first subphase focuses on identifying the most efficient PNE, following the approach presented in \citep{pradelski2012learning}. The second subphase centers on exploiting the gathered information to commit to the identified most efficient PNE.

\subsubsection{Exploration Phase} \label{sect:exploration-phase}
During this phase, each player pulls each arm $c_1$ times to estimate the expected reward $\hat{\mu}_{j,k}$ for each arm $k$. This estimate is calculated as the average of the observed rewards when pulling the arm. To further reduce regret during this phase, a collision-minimizing strategy is employed: at each round $t$, player $j$ selects the arm using the formula $(j + t) \mod K + 1$. This strategy aims to minimize the number of collisions among players.

\subsubsection{Learning and Exploitation Phase} \label{sect:learning-and-exploitation-phase}

\begin{figure}
    \centering
    \includegraphics[width=0.6\linewidth]{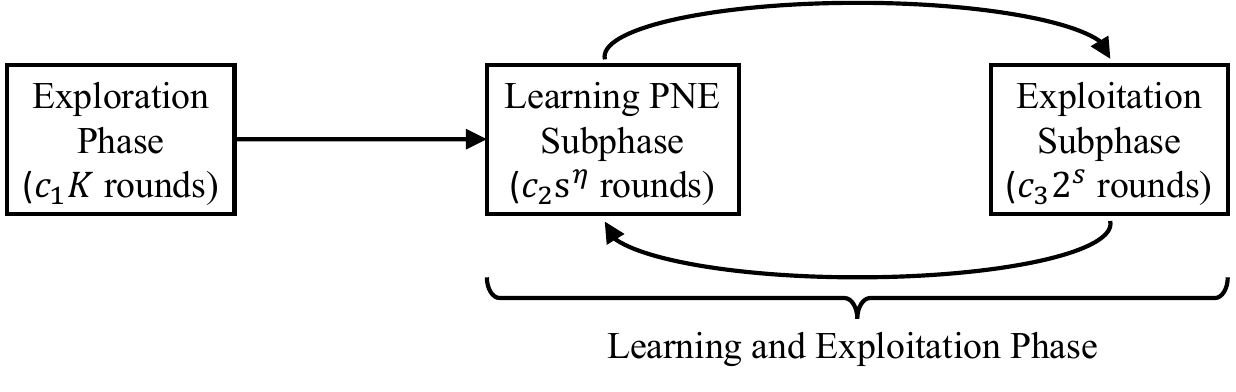}
    \vspace{-10px}
    \caption{Graphical overview of our algorithm: Utilizes hyper-parameters $c_1$, $c_2$, $c_3$, and $\eta$, with counter $s$ progressing through natural numbers. Initially, the algorithm enters an Exploration Phase for $c_1K$ rounds (\cref{sect:exploration-phase}), followed by a Learning and Exploitation phase (\cref{sect:learning-and-exploitation-phase}), where it alternates between Learning PNE and Exploitation Subphases based on the counter $s$. Specifically, in round $s$, the Learning PNE Subphase spans $c_2s^\eta$ rounds, and the Exploitation Subphase spans $c_32^s$ rounds.}
    \label{fig:algorithm}
\end{figure}

We next move on to the Learning and Exploitation Phase, which cyclically alternates between the Learning PNE and Exploitation Subphases.

Although our ultimate goal is to identify the most efficient Nash equilibrium in the game $\calG$ as outlined in \cref{defn:game}, precise parameter estimation for $\calG$ is unfeasible. Players can only approximate the game based on their arm reward estimates $\hat{\mu}_{j,k}$. Additionally, due to the potential non-uniqueness of the most efficient Nash equilibria in $\calG $, we introduce a randomized perturbation term $\gamma_{j,k}$ to each player's utility. Specifically, we aim to find the most efficient PNE in a perturbed game $\calG'$, which is given as:
\begin{equation} \label{eq:game-formulation-perturbed}
    \begin{aligned}
        & \calG' = \left(\calS, \left\{\calU_j'(\bfpi)\right\}_{j=1}^N\right), \quad \text{where} \\
        & \calS = [K] \text{ and } \calU_j'(\bfpi) = \frac{\hat{\mu}_{j,k} w_{j,k}}{\sum_{j'=1}^N \bbI[\pi_{j'}=k]w_{j',k}}+\gamma_{j,k}, k = \pi_j.
    \end{aligned}
\end{equation}
Here, $\gamma_{j,k}$ is randomly sampled from a uniform distribution over the interval $[0, \Gamma]$. As proven in ~\cref{thrm:game-perturbed}, under certain conditions, $\calG'$ has a unique most efficient PNE with probability 1, which also serves as a most efficient PNE in the original game $\calG$. Now we introduce the details of Learning PNE and Exploitation Subphases.

\paragraph{Learning PNE Subphase.} In this subphase, we adopt the framework proposed in \citep{pradelski2012learning} to identify the most efficient Nash equilibrium. Specifically, each player maintains a status tuple $(m_j, a_j, u_j)$ at each round $t$ during this phase. Here, $m_j$ can be one of four moods: discontent ($\discontent$), content ($\content$), watchful ($\watchful$), and hopeful ($\hopeful$). $a_j$ denotes the arm currently being played by player $j$, and $u_j$ represents the utility player $j$ gains when playing arm $a_j$.

Intuitively, a player's mood reflects their attitude toward their current strategy, similar to how individuals respond emotionally or cognitively to decisions. A discontent player, uncertain about the best choice, explores alternatives, akin to someone reconsidering their options when dissatisfied. A content player, confident in their strategy, sticks with it, much like someone satisfied with their decisions. A watchful player, cautious about their strategy’s effectiveness, may shift to a discontent mood if doubts arise. Finally, a hopeful player, believing in the potential for better returns, persists with their current strategy, similar to someone holding on to a risky investment in hopes of future gains.

Following the principles underlying mood design, at each round $t$, players select arms according to \cref{rule:pull-arm}, based on their mood from the previous round.
\begin{criterion}[Rule to pull arm] \label{rule:pull-arm}
    The chosen arm $\pi_j(t)$ is chosen based on the mood in the last round.
    \begin{itemize}
        \item A discontent ($\discontent$) player chooses arm $\pi_j(t)$ uniformly at random from $[K]$.
        \item A content ($\content$) player chooses arm $\pi_j(t) = a_j$ with probability $1-\epsilon$ and chooses all other arms with probability $\epsilon / (K-1)$.
        \item A watchful ($\watchful$) or hopeful ($\hopeful$) player chooses arm $\pi_j(t) = a_j$.
    \end{itemize}
\end{criterion}

\begin{remark}
    \cref{rule:pull-arm} encapsulates the different player moods. Discontent players, uncertain about the best strategy, explore the arm space uniformly, seeking alternatives. Content players, confident in their strategy, mostly stick with it but occasionally explore other options out of curiosity. Watchful and hopeful players, still assessing their strategy's effectiveness, continue with their current choice, adopting a cautious or optimistic wait-and-see approach.
\end{remark}

After pulling the selected arm $\pi_j(t)$, player $j$ observes his own reward $R_j(t)$ and the total reward $X_k(t)$ for the arm, where $k = \pi_j(t)$. The utility is then calculated as
$$
    u_j' = R_j(t)\hat{\mu}_{j,k}/X_k(t) + \gamma_{j,k}.
$$
The player's status is updated based on this computed utility $u_j'$, in accordance with \cref{rule:transit}.

\begin{criterion}[Rule to update status] \label{rule:transit}
    Let $0 < \epsilon < 1$ be a constant. Define the functions $F(u)$ and $G(u)$ as follows:
    \begin{equation} \label{eq:F-and-G}
        F(u)=-\frac{u}{(4+3\Gamma)N}+\frac{1}{3N} \quad \text{and} \quad G(u)=-\frac{u}{4+3\Gamma}+\frac{1}{3}.
    \end{equation}
    The status update is outlined below:
    \begin{itemize}
        \item A discontent ($\discontent$) player transits to status $(m_j, a_j, u_j) \leftarrow(\content$, $\pi_j(t), u'_j)$ with probability $\epsilon^{F(u'_j)}$ and remains unchanged otherwise.
        \item A content ($\content$) player's rule is based on his own choice $\pi_j(t)$. If $\pi_j(t) = a_j$, we have
        $$
            (m_j, a_j, u_j) \leftarrow
            \begin{cases}
                (\hopeful, a_j, u_j) & \text{if } u'_j > u_j, \\
                (\content, a_j, u_j) & \text{if } u'_j = u_j, \\
                (\watchful, a_j, u_j) & \text{if } u'_j < u_j.
            \end{cases}
        $$
        And under the case when $\pi_j(t) \ne a_j$, we have that if $u'_j > u_j$ then $(m_j, a_j, u_j) \leftarrow (\content, \pi_j(t), u'_j)$ with probability $\epsilon^{G(u'_j - u_j)}$ and remains unchanged with probability $1 - \epsilon^{G(u'_j - u_j)}$. If $u'_j \le u_j$, $(m_j, a_j, u_j)$ remains unchanged.
        \item A watchful ($\watchful$) player follows
        $$
            (m_j, a_j, u_j) \leftarrow
            \begin{cases}
                (\hopeful, a_j, u_j) & \text{if } u'_j > u_j, \\
                (\content, a_j, u_j) & \text{if } u'_j = u_j, \\
                (\discontent, a_j, u_j) & \text{if } u'_j < u_j.
            \end{cases}
        $$
        \item A hopeful ($\hopeful$) player follows
        $$
            (m_j, a_j, u_j) \leftarrow
            \begin{cases}
                (\content, a_j, u'_j) & \text{if } u'_j > u_j, \\
                (\content, a_j, u_j) & \text{if } u'_j = u_j, \\
                (\watchful, a_j, u_j) & \text{if } u'_j < u_j.
            \end{cases}
        $$
    \end{itemize}
\end{criterion}

\begin{remark}
    \cref{rule:transit} reflects the intuition behind player moods. Discontent players transition to content with some probability if their chosen strategy yields high utility. The design of $F(u)$ and $\epsilon$ ensures that $\epsilon^{F(u_j')}$ increases with $u_j'$, meaning higher utility raises the likelihood of selecting that strategy.
    
    Content players update their status based on the observed utility for the current round unless they have already explored other strategies, in which case they transition to a new strategy probabilistically based on utility. The design of $G(u)$ and $\epsilon$ ensures that $\epsilon^{G(u_j' - u_j)}$ increases with $u_j'$, meaning higher utility raises the probability of switching strategies. Unlike discontent players, content players compare new utility to their previous utility, as captured by the $u_j' - u_j$ term in the probability.

    Watchful and hopeful players adjust their status based on observed utility. A watchful player becomes discontent if their strategy has lower utility for two consecutive rounds, while a hopeful player, upon observing utility improvement for two consecutive rounds, remains with the current strategy and transitions to content.

    Consistent with existing work \citep{pradelski2012learning,young2009learning,marden2014achieving}, \cref{rule:pull-arm,rule:transit} follows the principles of mood design, forming a trial-and-error learning process similar to common human learning routines \citep{young2009learning}.
\end{remark}

Additionally, in the Learning PNE Subphase, each player maintains a counter $Q_{j,k}$ for each arm $k \in [K]$, tracking how often arm $k$ is chosen while the player is in a Content mood.

\paragraph{Exploitation Subphase.}
During the Exploitation Subphases, each player chooses the arm with the maximal number of times that arm is chosen with a content mood, \textit{i.e.}, $k = \argmax_{k'}Q_{j,k'}$. Then the player will pull arm $k$ for the remaining rounds in the subphase.

\subsection{Theoretical Analysis}
In this subsection, we analyze the performance of our online algorithm. We first establish a metric to quantify the challenge of identifying the most efficient PNE in the presence of imprecise resource payoffs. Subsequently, this metric serves as a basis for deriving rigorous bounds on both regret and the number of rounds in which players do not follow the most efficient PNE.

\subsubsection{Quantifying the Challenge of Identifying the Most Efficient PNE}
We introduce the sets $A^{\PNEBound}$ and $A^{\NOPNEBound}$ to denote Nash equilibria and non-equilibria strategy profiles, respectively. Let
\begin{equation} \label{eq:delta}
    \begin{aligned}
        \delta^{\PNEBound} & = \min\big\{\calU_j(\bfpi) - \calU_j(\pi', \bfpi_{-j}): \\
        & \quad \quad \quad \,\,\, \bfpi \in A^{\PNEBound}, j \in [N], \pi' \in [K], \pi' \ne \pi_j\big\}. \\
        \delta^{\NOPNEBound} & = \sup\left\{\delta': \forall \bfpi \in A^{\NOPNEBound}, \bfpi \text{ is not a } \delta'\text{-Nash equilibrium}\right\}. \\
        \delta & = \min\left\{\delta^{\PNEBound}, \delta^{\NOPNEBound}\right\}.
    \end{aligned}
\end{equation}
Here, $\delta^\PNEBound$ quantifies the probability of erroneously disregarding a PNE $\mathbf{\pi}$ due to inaccuracies in reward estimation. Conversely, $\delta^\NOPNEBound$ gauges the risk of misidentifying a non-equilibrium strategy profile $\mathbf{\pi}$ as a PNE. To capture both of these risks, we set $\delta$ as the minimum of $\delta^\PNEBound$ and $\delta^\NOPNEBound$. We further show from the following example that both $\delta^{\PNEBound}$ and $\delta^{\NOPNEBound}$ may determine $\delta$.

\begin{example} \label{example:pne-and-nopne}
    Set $N = 3$ and $K = 2$. Here we suppose that each player has the same weights on the resources (as per Condition (2) in \cref{thrm:pne-exists-special}).
    When arms' rewards are $[1, 0.7]$ and players' weights are $[1, 0.8, 0.4]$, we have $0.052 = \delta^{\PNEBound} > \delta^{\NOPNEBound} = 0.026$. By contrast, when arms' rewards are $[1, 0.6]$ and players' weights are $[1, 2/3, 4/9]$, we have $0.040 = \delta^{\PNEBound} < \delta^{\NOPNEBound} = 0.068$. Details of the two cases can be found in \cref{tab:case-pne-larger,tab:case-nopne-larger} in \cref{app:examples}.
\end{example}

Let $\Delta$ be the largest discrepancy between the estimated and true rewards $\Delta = \max_{j,k}\left|\hat{\mu}_{j,k} - \mu_k\right|$. We have the following theorem.
\begin{theorem} \label{thrm:game-perturbed}
    Suppose $\Delta < \delta/(4K)$, $\Gamma \le \delta / (4N)$, and the PNE exists for the game $\calG$ in \cref{defn:game}. Then the PNE exists for the game $\calG'$ (\cref{eq:game-formulation-perturbed}). In addition, the most efficient PNE in the game $\calG'$  is unique with probability $1$ and the equilibrium is a most efficient PNE in the original game $\calG$.
\end{theorem}
\begin{remark}
    \cref{thrm:game-perturbed} demonstrates that under bounded estimation error $\Delta$ and perturbation $\Gamma$, the most efficient PNE in $\mathcal{G}$ can be uniquely identified.
\end{remark}

\subsubsection{Performance Analysis}
\begin{figure*}[t]
    \centering
    \includegraphics[width=\linewidth]{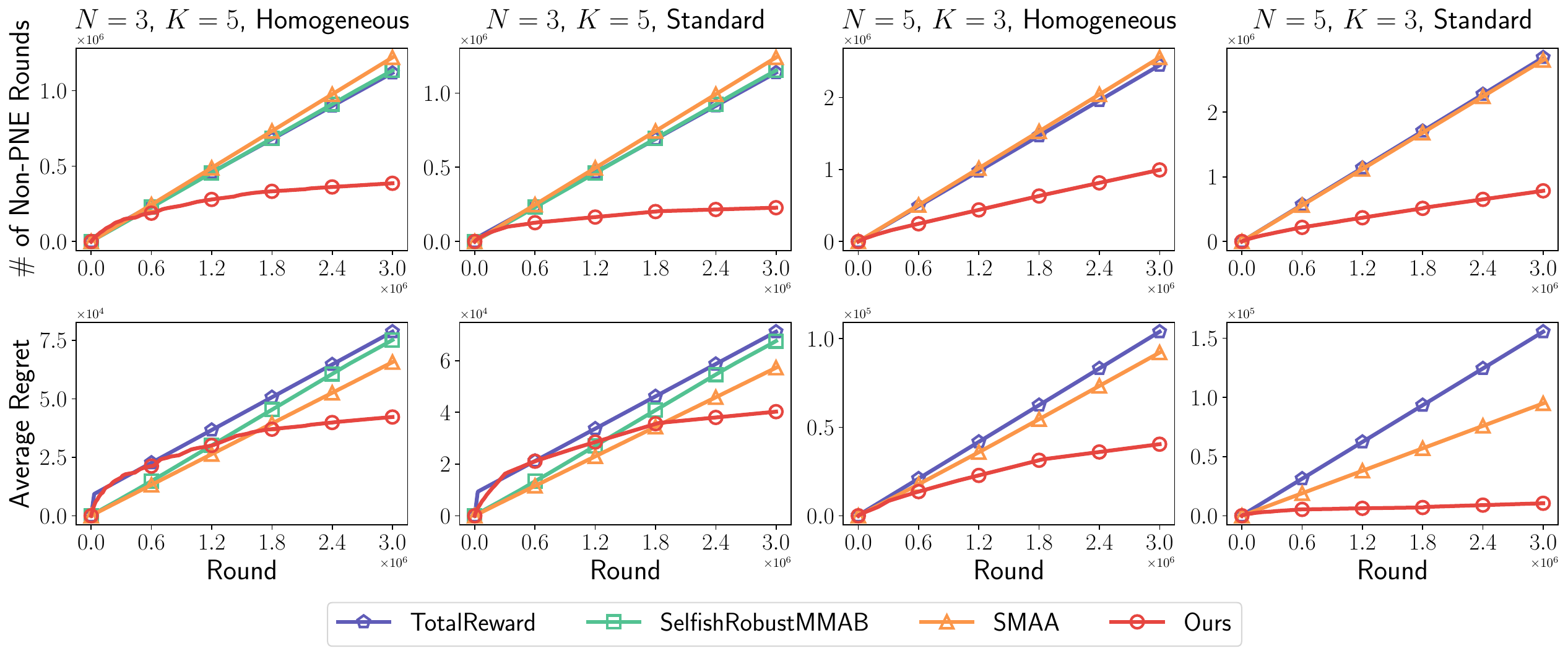}
    \vspace{-10px}
    \caption{Curves showing the average regret and the number of rounds where players do not follow the most efficient PNE. Note that SelfishRobustMMAB could not be applied to scenarios when $N > K$.}
    \label{fig:exp}
\end{figure*}

We now analyze the two metrics introduced in \cref{sect:metrics}. Following \citep{bistritz2018distributed,bistritz2021game}, the Learning PNE Subphase defines a Markov process $\calM(\epsilon; \calG')$ on the states of all players $Z = \{(m_j, a_j, u_j)\}_{j \in [N]}$ based on the parameter $\epsilon$ and the game $\calG'$. Based on this observation, we could obtain the following results.
Note that $c_1, c_2, c_3, \eta$ are hyper-parameters of our algorithm.

\begin{theorem} \label{thrm:regret}
    Suppose $\Delta < \delta/(4K)$, $\Gamma \le \delta / (4N)$, $w_{j, k} > 0$ for all $j \in [N]$ and $k \in [K]$, and the PNE exists for the game $\calG$ in \cref{defn:game}. Set $c_1 \ge 8K^2\log T/\delta^2$. There exists a small enough $\epsilon > 0$ such that when $T$ is large enough, we have that $\forall j \in [N]$,
    $$
        \begin{aligned}
            \bbE[\regret_j(T)] & \le O\left(\mu_{\max}\left(c_1 + c_2\log^{1 + \eta} T + c_3\log T\right)\right), \\
             \bbE[\noneq(T)] & \le O\left(c_1 + c_2\log^{1 + \eta} T + c_3\log T\right),
        \end{aligned}
    $$
    where $\mu_{\max} = \max \mu_k$.
\end{theorem}

\begin{remark}
    The theorem establishes that our algorithm converges to the most efficient PNE in the long term. It further characterizes the regret for each player as $O(\log^{1+\eta} T)$ for any $\eta > 0$, under sufficient conditions. This aligns with existing results in \citep{bistritz2021game}. Additionally, as $\eta \to 0$, our bounds reduce to $O(\log T)$, consistent with standard regret bounds in the multi-armed bandit literature \citep{boursier2022survey,slivkins2019introduction}.

    In addition, we emphasize that the minimum value of $T$ required to guarantee the bounds in \cref{thrm:regret} implicitly depends on all game parameters and hyperparameters, as is also the case in~\citep{bistritz2021game}. Notably, the inequalities in \cref{thrm:regret} do not imply that smaller values of $c_1$, $c_2$, or $c_3$ necessarily lead to better performance. This dependence motivates our empirical tuning of these parameters in the experiments (\cref{sect:experiments}).
\end{remark}

\section{Synthetic Experiments} \label{sect:experiments}

In this section, we validate the effectiveness of our online method through synthetic experiments.

\subsection{Data-generating Process}
In our experiments, we assess two distinct scenarios: $(N, K) = (3, 5)$ and $(N, K) = (5, 3)$, effectively examining both $N < K$ and $N > K$ cases. We set the time horizon, $T$, to 3,000,000 for each scenario. For modeling the payoff distribution of each resource, we use the beta distribution, constrained within the interval $[0, 1]$. Specifically, for each resource $k$, we independently draw two parameters, $\alpha_k$ and $\beta_k$, from a uniform distribution over $[0, 10]$. The rewards for each resource (or arm) in different rounds are then sampled from $\xi_k = \text{Beta}(\alpha_k, \beta_k)$. The players' weights $w_{j,k}$ are drawn from $\text{Uniform}(0, 1)$. 

We further explore two specific scenarios in the context of players' weights: (1) \textit{Homogeneous setting}, where all players have identical weights across various resources, aligning with the second condition in \cref{thrm:pne-exists-special}; and (2) \textit{Standard setting}, where players' weights across resources can differ without constraints.

\subsection{Baselines and Hyper-parameters}

We benchmark our approach against the following baselines.
\begin{itemize}
    \item \textbf{TotalReward}: Proposed by \citet{bande2019multi}, this method emphasizes a shareable reward case, aiming to maximize the collective reward for all players. In scenarios where $N < K$, the optimal strategy for each round is to have all players select the top $N$ arms. When $N \ge K$, every arm should be selected in every round. This strategy operates in an explore-then-commit manner. Specifically, during the initial $\alpha \log T$ rounds, the algorithm explores arms randomly. Thereafter, it commits to the optimal strategy in each round. The hyper-parameter $\alpha$ is chosen from the set $\{100, 200, 500, 1000, 2000\}$.
    \item \textbf{SelfishRobustMMAB}: Proposed by \citet{boursier2020selfish}, this algorithm aims to find the most efficient PNE. However, it assumes that when multiple players select the same resource, none receive a payoff. Under this context, the method provides algorithms resistant to the deviations of individual selfish players. The algorithm utilizes KL-UCB and a coefficient $\beta$ to determine the exploration of sub-optimal arms. The hyper-parameter $\beta$ is examined over the range $\{0.01, 0.02, 0.05, 0.1, 0.2, 0.5\}$.
    \item \textbf{SMAA}: According to \citet{xu2023competing}, this approach also aims to identify the PNE. It assumes an averaging allocation model, meaning players equally divide the payoffs of resources when selecting the same one. Similar to \citep{boursier2020selfish}, a hyper-parameter $\beta$ is utilized to modulate the exploration of sub-optimal arms. We explore values for $\beta$ in the set $\{0.01, 0.02, 0.05, 0.1, 0.2, 0.5\}$.
\end{itemize}

As for our method, we slightly abuse the notation for clarity. During the exploration phase, each arm is explored for a duration of $c_1 K^2 \log T / \delta^2$. The value of $\Gamma$ is set to $\delta / (4N)$ and $c_3$ is set to $200$. Our search encompasses various hyper-parameters, including $c_1$ from $\{0.001, 0.01\}$, $c_2$ from $\{1000, 5000\}$, $\eta$ from $\{1, 2\}$, and $\epsilon$ from $\{0.001, 0.003, 0.005\}$.

\subsection{Analysis}
We conducted experiments across 50 unique simulations, repeatedly sampling the resources' payoff distributions and players' weights. Our method, along with baseline approaches, was assessed by measuring the average regret (given in \cref{eq:regret}) across all players and tracking the number of rounds in which players deviated from the most efficient PNE (given in \cref{eq:non-eq}). The experimental results are presented in \cref{fig:exp}.

From the figure, we can observe that our methodology consistently surpasses all baselines across different data-generating processes. Given that none of the baselines is designed specifically for our PPA-Game, they often fail to identify the PNE. Consequently, their performance curves are linear across all setups, indicating that both regrets and the count of Non-PNE rounds increase linearly with $T$. By contrast, our method exhibits a diminishing slope as $T$ grows, reinforcing the idea that our approach eventually converges to the most efficient PNE. These experimental findings also demonstrate that the PNE solution can vary significantly based on game formulations, emphasizing the pivotal role of precise game formulation in practical scenarios.


\section{Conclusions} \label{sect:conclusion}
In summary, we introduce the PPA-Game as a new framework for analyzing online content creator competition, investigating the properties of PNE. We extend this to the Multi-player Multi-Armed Bandit framework, studying creators' adaptive strategies in dynamic demand settings. Additionally, we propose and validate a new online algorithm for content optimization, backed by theoretical guarantees and empirical evidence, demonstrating its effectiveness in enhancing long-term creator performance.
\section*{Acknowledgements}
We thank Alex Suzdaltsev, Grigorii Siialov, and Iakov Zhukov for identifying and correcting an error in an earlier version of Example~\ref{example:pne}.

\bibliographystyle{plainnat}
\bibliography{references}

\clearpage

\appendix

\section{Further Analysis on Pure Nash Equilibrium (PNE)}
\subsection{The Non-existence of PNE in Certain Scenarios} \label{sect:app-non-existence}
To show the fact that a Pure Nash Equilibrium (PNE) might not always exist, we present a specific counter-example below\footnote{We thank Alex Suzdaltsev, Grigorii Siialov, and Iakov Zhukov for pointing out that the game specified in an earlier version of this example admitted a PNE and for helping us correct the parameters. The payoff vector and weight matrix have been revised to match the improvement cycle in \cref{fig:example-pne-non}.}.

\begin{example}\label{example:pne}
    Consider $N = 7$ players and $K = 4$ resources, labeled $A$, $B$, $C$, and $D$. The resource payoff vector is $\bfmu = (1, 0.4904, 0.1065, 0.2447)$, and the player-resource weight matrix is given in \cref{eq:counter-example}, with rows corresponding to players $1,\ldots,7$ and columns to resources $A,B,C,D$. Exhaustive enumeration of all $4^7 = 16{,}384$ strategy profiles confirms that this game has no PNE.
\end{example}

\begin{figure}[h]
    \centering
    \begin{minipage}{0.45\linewidth}
        \centering
        \includegraphics[height=0.4\linewidth]{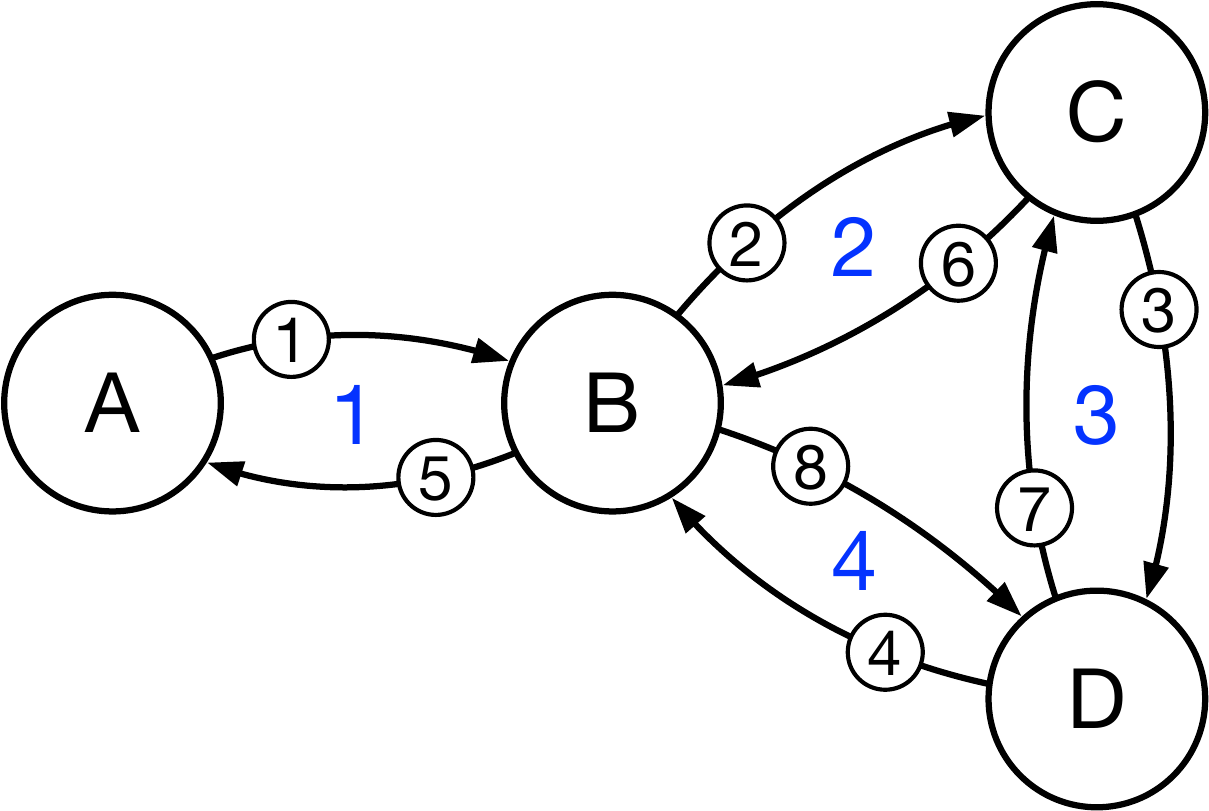}
        \caption{An improvement cycle in a game with no PNE.}
        \label{fig:example-pne-non}
    \end{minipage}
    \hfill
    \begin{minipage}{0.45\linewidth}
        \begin{equation}\label{eq:counter-example}
            \bfW = \begin{pmatrix}
                0.1657 & 0.5755 & 0 & 0 \\
                0 & 0.4057 & 1 & 0 \\
                0 & 0 & 0.0314 & 0.5623 \\
                0 & 0.4217 & 0 & 0.9517 \\
                1 & 0 & 0 & 0 \\
                0 & 1 & 0 & 0 \\
                0 & 0 & 0 & 1
            \end{pmatrix}.
        \end{equation}
    \end{minipage}
\end{figure}

\Cref{fig:example-pne-non} illustrates a strict improvement cycle starting from $\pi^{(0)} = (A,B,C,D,A,B,D)$, where the entries specify the resources chosen by players $1,\ldots,7$, respectively. Players $5$, $6$, and $7$ remain on resources $A$, $B$, and $D$ throughout the cycle. The large circles denote resources, the blue numbers identify the moving players, and the circled numbers indicate the order of the moves.

In rounds $1$--$4$, players $1$, $2$, $3$, and $4$ make the respective deviations $A \to B$, $B \to C$, $C \to D$, and $D \to B$. In rounds $5$--$8$, the same players move again in the same order, making the respective deviations $B \to A$, $C \to B$, $D \to C$, and $B \to D$. Each move strictly increases the moving player's payoff. The eighth move restores $\pi^{(0)}$, completing the improvement cycle.

\subsection{High Probability of PNE Existence} \label{sect:app-high-probability}
We experiment with cases where the values of $N$ and $K$ are large, using Monte Carlo sampling. We randomly select values for $N$ and $K$, uniformly from $[10, 100]$. Moreover, for each player-resource pair $(j,k)$, with $j \in [N]$ and $k \in [K]$, weights $w_{j,k}$ and payoffs $\mu_k$ are sampled from four different types of distributions.
\begin{itemize}
\item \textbf{Uniform Distribution}: We sample payoffs and weights using a uniform distribution $U(0, 1)$.
\item \textbf{Truncated Normal Distribution}: We sample payoffs and weights using a normal distribution $N(0.5, 0.1^2)$ and truncate them into $[0, 1]$. This distribution provides a scenario where values are concentrated around the mean of $0.5$.
\item \textbf{Truncated T-Distribution}: We utilize a t-distribution with 5 degrees of freedom, adjusted by adding $0.5$ to center the mean at $[0, 1]$. Similarly, we truncate the values into $[0, 1]$. This distribution also provides a scenario where values are concentrated around the mean of $0.5$.
\item \textbf{Beta Distribution}: We opt for a beta distribution with parameters $\alpha = \beta = 0.5$. This distribution is skewed towards the extremes of $0$ and $1$, representing scenarios where values are highly polarized.
\end{itemize}
For every randomly generated game configuration, considering the computational complexity of evaluating all strategy profiles, we trace an improvement path starting from the strategy profile $\bfpi(0) = (0, 0, \dots, 0)$. If this improvement path concludes within finite steps, it implies the existence of a PNE. Out of the $5,000,000$ randomly generated game configurations in each distribution (\textit{i.e.}, totaling $20,000,000$ randomly generated configurations), only one exhibited an infinite improvement path. This suggests the existence of a PNE in the remaining $19,999,999$ configurations. A comprehensive evaluation of strategy profiles confirmed the non-existence of a PNE in the only exceptional case, concisely illustrated in \cref{sect:app-non-existence}.



\section{Details About \cref{example:pne-and-nopne}} \label{app:examples}
The details of toy cases in \cref{example:pne-and-nopne} can be found in \cref{tab:case-pne-larger,tab:case-nopne-larger}.
\begin{table}[th]
    \centering
    \caption{A showcase that $\delta^{\PNEBound} > \delta^{\NOPNEBound}$. Set $N = 3$ and $K = 2$. Assume the resources' payoffs to be $[1, 0.7]$. Suppose each player has the same weights on the resources (as per Condition (2) in \cref{thrm:pne-exists-special}), with weights $[1, 0.8, 0.4]$, respectively. We list all strategy profiles and find that $\delta^{\PNEBound} = \min\{0.052, 0.133\} = 0.052 > 0.022 = \min\{0.518, 0.022, 0.052, 0.133, 0.356, 0.873\} = \delta^{\NOPNEBound}$.}
    \resizebox{\linewidth}{!}{
    \begin{tabular}{c|c|c|c}
        \toprule
        Strategy profile & PNE? & Least reward change by deviation if PNE & $\delta'$-Nash equilibrium if not PNE \\
        \midrule
        $(1, 1, 1)$ & \xmark & & $0.518$. Deviated strategy $(1, 1, 2)$ \\
        $(1, 1, 2)$ & \xmark & & $0.022$. Deviated strategy $(1, 2, 2)$ \\
        $(1, 2, 1)$ & \cmark & $0.052$. Deviated strategy $(1, 2, 2)$ \\
        $(1, 2, 2)$ & \xmark & & $0.052$. Deviated strategy $(1, 2, 1)$ \\
        $(2, 1, 1)$ & \cmark & $0.133$. Deviated strategy $(2, 1, 2)$ \\
        $(2, 1, 2)$ & \xmark & & $0.133$. Deviated strategy $(2, 1, 1)$ \\
        $(2, 2, 1)$ & \xmark & & $0.356$. Deviated strategy $(2, 1, 1)$ \\
        $(2, 2, 2)$ & \xmark & & $0.873$. Deviated strategy $(2, 2, 1)$ \\
        \bottomrule
    \end{tabular}}
    \label{tab:case-pne-larger}
\end{table}

\begin{table}[th]
    \centering
    \caption{A showcase that $\delta^{\PNEBound} < \delta^{\NOPNEBound}$. Set $N = 3$ and $K = 2$. Assume the resources' payoffs to be $[1, 0.6]$. Suppose each player has the same weights on the resources (as per Condition (2) in \cref{thrm:pne-exists-special}), with weights $[1, 2 / 3, 4/9]$, respectively. We list all strategy profiles and find that $\delta^{\PNEBound} = \min\{0.040, 0.068, 0.126\} = 0.040 < 0.068 = \min\{0.389, 0.068, 0.215, 0.360, 0.874\} = \delta^{\NOPNEBound}$.}
    \resizebox{\linewidth}{!}{
    \begin{tabular}{c|c|c|c}
        \toprule
        Strategy profile & PNE? & Least reward change by deviation if PNE & $\delta'$-Nash equilibrium if not PNE \\
        \midrule
        $(1, 1, 1)$ & \xmark & & $0.389$. Deviated strategy $(1, 1, 2)$ \\
        $(1, 1, 2)$ & \cmark & $0.040$. Deviated strategy $(1, 2, 2)$ \\
        $(1, 2, 1)$ & \cmark & $0.068$. Deviated strategy $(1, 2, 2)$ \\
        $(1, 2, 2)$ & \xmark & & $0.068$. Deviated strategy $(1, 2, 1)$ \\
        $(2, 1, 1)$ & \cmark & $0.126$. Deviated strategy $(1, 1, 1)$ \\
        $(2, 1, 2)$ & \xmark & & $0.215$. Deviated strategy $(2, 1, 1)$ \\
        $(2, 2, 1)$ & \xmark & & $0.360$. Deviated strategy $(2, 1, 1)$ \\
        $(2, 2, 2)$ & \xmark & & $0.874$. Deviated strategy $(2, 2, 1)$ \\
        \bottomrule
    \end{tabular}}
    \label{tab:case-nopne-larger}
\end{table}

\section{Omitted Proofs}
\subsection{Proof of \cref{thrm:pne-exists}} \label{sect:proof-pne-exists}
We first introduce some notations. For a strategy profile $\bfpi$, define $W_k(\bfpi)$ as the sum of weights of players that choose resource $k$, \textit{i.e.},
\begin{equation*}
    W_k(\bfpi) = \sum_{j=1}^N \bbI[\pi_j = k]w_{j,k}.
\end{equation*}
Furthermore, Let $\ell = (\bfpi(1), \bfpi(2), \dots, \bfpi(t), \bfpi(t + 1) = \bfpi(1))$ be an improvement cycle. For any $s$ such that $\bfpi(s+1)$ is an improvement step of $\bfpi(s)$, define $j(s)$ be the deviated player in this step and the new strategy is $k(s)$, \textit{i.e.}, $\bfpi(s+1) = \left(k(s), \bfpi(s)_{-j(s)}\right)$.
Besides, we let $Q_{k}(\ell)$ be the minimal sum of players' weights on the resource $k$ in the improvement cycle $\ell$, \textit{i.e.},
\begin{equation*}
    Q_{k}(\ell) = \min_{\bfpi(s) \in \ell}{\sum_{j \in B_{k}(\bfpi(s))}w_{j,k}} = \min_{\bfpi(s) \in \ell}W_k(\bfpi(s)).
\end{equation*}

We first introduce the following lemmas.

\begin{lemma} \label{lemma:proof-pne-exists-1}
    Let $\ell = (\bfpi(1), \bfpi(2), \dots, \bfpi(t), \bfpi(t + 1) = \bfpi(1))$ be an improvement cycle. For any $1 \le s \le t$, let $j = j(s)$ and $k = k(s)$. Then we have
    \begin{equation*} \label{eq:lemma-smaller-than-Qk}
        \calU_j(\bfpi(s + 1)) \le  \frac{\mu_kw_{j,k}}{Q_{k}(\ell) + w_{j,k}}.
    \end{equation*}
\end{lemma}

\begin{proof}
    Since $\bfpi(s + 1)$ is an improvement step of $\bfpi(s)$, player $j$ earns a larger payoff by deviation in this step, leading to the fact that $w_{j, k} > 0$ and the right-hand side of \cref{eq:lemma-smaller-than-Qk} is well-defined. Furthermore, by the definition of $Q_k(\ell)$, we have
    \begin{equation*}
        W_k(\bfpi(s+1)) = W_k(\bfpi(s)) + w_{j,k} \ge Q_k(\ell) + w_{j,k}.
    \end{equation*}
    As a result,
    \begin{equation*}
        \calU_j(\bfpi(s+ 1)) = \frac{\mu_kw_{j,k}}{W_{k}(\bfpi(s+1))} \le \frac{\mu_kw_{j,k}}{Q_k(\ell) + w_{j,k}}.
    \end{equation*}
    Now the claim follows.
\end{proof}

\begin{lemma} \label{lemma:proof-pne-exists-2}
    Let $\ell = (\bfpi(1), \bfpi(2), \dots, \bfpi(t), \bfpi(t + 1) = \bfpi(1))$ be an improvement cycle. Let $k$ be any resource such that $\exists 1 \le s' \le t$, $k(s') = k$. Then there exist an index $s$ such that $\pi_{j(s)}(s) = k$ and
    \begin{equation*}
        \calU_{j}\left(\bfpi(s)\right)=\frac{\mu_kw_{j,k}}{Q_{k}(\ell) + w_{j,k}} \quad \text{where} \quad j = j(s).
    \end{equation*}
\end{lemma}

\begin{proof}
    Since $\exists 1 \le s' \le t$ such that $k(s') = k$, there must exists rounds when a player deviate from the resource $k$. As a result, there must exist a constant $1 \le s \le t$ such that
    \begin{equation*}
        \pi_{j(s)}(s) = k \quad \text{and} \quad Q_k(\ell) = W_k(\bfpi(s + 1)).
    \end{equation*}
    As a result, let $j = j(s)$ and we have
    \begin{equation*}
        \calU_j(\bfpi(s)) = \frac{\mu_k w_{j,k}}{W_k(\bfpi(s))} = \frac{\mu_k w_{j,k}}{W_k(\bfpi(s + 1)) + w_{j,k}} = \frac{\mu_k w_{j,k}}{Q_k(\ell) + w_{j,k}}.
    \end{equation*}
    Now the claim follows.
\end{proof}

\begin{lemma} \label{lemma:proof-pne-exists-3}
    Let $\ell = (\bfpi(1), \bfpi(2), \dots, \bfpi(t), \bfpi(t + 1) = \bfpi(1))$ be an improvement cycle. Then for every improvement step $1 \le s \le t$ and resource $k$ such that $k(s) = k$, there exist $1 \le s' \le t$ and $k' = k(s')$ such that
    \begin{equation*}
        \frac{\mu_{k}}{Q_{k}(\ell) + \epsilon_0} < \frac{\mu_{k'}}{Q_{k'}(\ell) + \epsilon_0}.
    \end{equation*}
\end{lemma}

\begin{proof}
    By \cref{lemma:proof-pne-exists-2}, we have that there exists an index $s'$ such that $j = j(s')$, $\pi_j(s') = k$, and
    \begin{equation*}
        \calU_j(\bfpi(s')) = \frac{\mu_kw_{j,k}}{Q_k(\ell) + w_{j,k}}.
    \end{equation*}
    Set $k' = k(s')$. Since $\bfpi(s'+1)$ is an improvement step of $\bfpi(s')$ for player $j = j(s')$, we have
    \begin{equation*} \label{eq:proof-lemma-proof-pne-exists-1-1}
        \frac{\mu_kw_{j,k}}{Q_k(\ell) + w_{j,k}}  = \calU_j(\bfpi(s')) < \calU_j(\bfpi(s' + 1)) \le \frac{\mu_{k'}w_{j,{k'}}}{Q_{k'}(\ell) + w_{j,k'}},
    \end{equation*}
    where the last inequality is due to \cref{lemma:proof-pne-exists-1}.

    First, we claim that $\mu_k \le \mu_{k'}$ by contradiction. Otherwise, we assume $\mu_k > \mu_{k'}$ and we obtain that $\mu_k/\mu_{k'} > N_{0}/\epsilon_0$ by the assumption. Based on \cref{eq:proof-lemma-proof-pne-exists-1-1}, we have:
    \begin{equation*}
        \frac{\mu_{k'}}{\mu_k} > \frac{w_{j,k}}{Q_{k}(\ell) + w_{j,k}} \left/ \frac{w_{j,k'}}{Q_{k'}(\ell) + w_{j,k'}}\right..
    \end{equation*}
    According to the assumption of $\epsilon_0$ and $N_0$, and acknowledging that the function $\frac{x}{x+Q}$ is increasing, we have:
    \begin{equation*}
        \frac{w_{j,k}}{Q_{k}(\ell) + w_{j,k}} \ge \frac{\epsilon_0}{N_0} \quad \text{and} \quad \frac{w_{j,k'}}{Q_{k'}(\ell) + w_{j,k'}} \leq 1.
    \end{equation*}
    As a result,
    \begin{equation*}
        \frac{\mu_{k'}}{\mu_k} > \frac{w_{j,k}}{Q_{k}(\ell) + w_{j,k}} \left/ \frac{w_{j,k'}}{Q_{k'}(\ell) + w_{j,k'}}\right. \ge \frac{\epsilon_0}{N_0} > \frac{\mu_{k'}}{\mu_k}.
    \end{equation*}
    This leads to a contradiction. As a result, we have $\mu_k \le \mu_{k'}$. Hence, by the assumption, we have
    \begin{equation*}
        \frac{\mu_{k'}}{\mu_k} > \frac{N_0}{\epsilon_0} \ge \frac{W_{k'}(\bfpi(s'+1))}{Q_{k}(\ell)+\epsilon_0} \ge \frac{Q_{k'}(\ell)+\epsilon_0}{Q_{k}(\ell)+\epsilon_0}.
    \end{equation*}
    Now the claim follows.
\end{proof}

With these lemmas, we can prove the original theorem.

\begin{proof}[Proof of \cref{thrm:pne-exists}]
    Suppose there exists an infinite improvement path. Then there must exist an improvement cycle
    \begin{equation*}
        \ell = \left(\bfpi(1), \bfpi(2), \dots, \bfpi(t), \bfpi(t+1) = \bfpi(1)\right).
    \end{equation*}
    Let $s_1 = 1$ and $k_1 = k(s_1)$. According to \cref{lemma:proof-pne-exists-3}, there exist $1 \le s_2 \le t$ and $k_2 = k(s_2)$ such that $\mu_{k_1} / (Q_{k_1}(\ell) + \epsilon_0) < \mu_{k_2} / (Q_{k_2}(\ell) + \epsilon_0)$. Continue this process for $t + 1$ times and we have $s_1, s_2, \dots, s_{t+1}$ and $k_1, k_2, \dots, k_{t+1}$ with $k_i = k(s_i)$ such that
    \begin{equation*} \label{eq:proof-pne-exists-1}
        \frac{\mu_{k_1}}{Q_{k_1}(\ell) + \epsilon_0} < \frac{\mu_{k_2}}{Q_{k_2}(\ell) + \epsilon_0} < \dots < \frac{\mu_{k_{t+1}}}{Q_{k_{t+1}}(\ell) + \epsilon_0}.
    \end{equation*}
    Since $k_i$ ($1 \le i \le t + 1$) can only take value in $\{1, 2, \dots, t\}$, there must exists two indices $i_1 \ne i_2$ such that $k_{i_1} = k_{i_2}$. This leads to a contradiction since \cref{eq:proof-pne-exists-1} suggests that $\forall 1 \le i_1, i_2 \le t + 1$, $k_{i_1} \ne k_{i_2}$. Now we can conclude that all improvement paths end in finite steps and the PNE exists.
\end{proof}

\subsection{Proof of \cref{thrm:pne-exists-special}} \label{sect:proof-pne-exists-special}
We share similar proof techniques for the three conditions. Generally speaking, we want to prove that every improvement path is finite, which leads to the existence of PNE. In addition, for a strategy profile $\bfpi$, define $W_k(\bfpi)$ as the sum of weights of players that choose resource $k$, \textit{i.e.},
\begin{equation*}
    W_k(\bfpi) = \sum_{j=1}^N \bbI[\pi_j = k]w_{j,k}.
\end{equation*}
For any $s$ such that $\bfpi(s+1)$ is an improvement step of $\bfpi(s)$, define $j(s)$ be the deviated player in this step and the new strategy is $k(s)$, \textit{i.e.}, $\bfpi(s+1) = \left(k(s), \bfpi(s)_{-j(s)}\right)$.

\begin{proof} [Proof of Condition (1)]
    Suppose there exists an infinite improvement path. Then there must exist an improvement cycle
    \begin{equation*}
        \left(\bfpi(1), \bfpi(2), \dots, \bfpi(t), \bfpi(t+1) = \bfpi(1)\right).
    \end{equation*}
    Since for all $j \in [N]$ and $k, k' \in [K]$, we have $w_{j,k} = w_{j,k'}$, we slightly abuse the notation here and use $w_j$ to denote $w_{j,k}$ for all $k \in [K]$. In addition, for any $1 \le s \le t$, define the potential function as
    \begin{equation*}
        \Phi(\bfpi(s)) = \sort \left(\frac{\mu_k}{W_k(\bfpi(s))}\right)_{k=1}^K.
    \end{equation*}
    Here $\Phi(\bfpi(s))$ is a sorted vector and each element in the vector corresponds  to a unit reward in a resource. Note that when $W_k(\bfpi(s)) = \sum_{j=1}^N\bbI[\pi_j(s)=k]w_{j} = 0$ for some $s$ and $k$, we set the value $\mu_k / W_k(\bfpi(s))$ be infinity ($\infty$). We suppose that $a < \infty$ for any real number $a$ and $\infty = \infty$.

    For each $1 \le s \le t$, since $\bfpi(s+1) = \left(k(s), \bfpi_{-j(s)}(s)\right)$ is an improvement step of $\bfpi(s)$ and $\mu_{\pi_{j(s)}(s)} = \mu_{k(s)}$ by condition, setting $j = j(s)$ for brevity, we have
    \begin{equation*} \label{eq:proof-proposition-pne-2}
        \frac{\mu_{\pi_j(s+1)}}{W_{\pi_j(s+1)}(\bfpi(s + 1))} > \frac{\mu_{\pi_j(s)}}{W_{\pi_j(s)}(\bfpi(s))}.
    \end{equation*}
    In addition, note that player $j$ deviate from resource $\pi_j(s)$ to $\pi_j(s+1)$ and we have that $W_{\pi_j(s)}(\bfpi(s)) \ge W_{\pi_j(s)}(\bfpi(s + 1))$ and $W_{\pi_j(s+1)}(\bfpi(s)) < W_{\pi_j(s+1)}(\bfpi(s + 1))$. Therefore,
    \begin{equation*} \label{eq:proof-proposition-pne-3}
        \frac{\mu_{\pi_j(s)}}{W_{\pi_j(s)}(\bfpi(s + 1))} \ge \frac{\mu_{\pi_j(s)}}{W_{\pi_j(s)}(\bfpi(s))} \quad \text{and} \quad \frac{\mu_{\pi_j(s+1)}}{W_{\pi_j(s+1)}(\bfpi(s))} > \frac{\mu_{\pi_j(s+1)}}{W_{\pi_j(s+1)}(\bfpi(s + 1))}.
    \end{equation*}
    Since $\Phi(\bfpi(s + 1))$ differs from $\Phi(\bfpi(s))$ only on the items \textit{w.r.t.} $k = \pi_j(s)$ and $k = \pi_j(s + 1)$, by examining \cref{eq:proof-proposition-pne-2} and \cref{eq:proof-proposition-pne-3}, it is easy to check that $\Phi(\bfpi(s + 1))$ is greater than $\Phi(\bfpi(s))$ by lexicographic order. As a result,
    \begin{equation*}
        \Phi(\bfpi(1)) = \Phi(\bfpi(t + 1)) \succ \Phi(\bfpi(t)) \succ \dots \succ \Phi(\bfpi(1)),
    \end{equation*}
    where $\succ$ represents "greater than by lexicographic order". This leads to a contradiction.
\end{proof}

\begin{proof} [Proof of Condition (2)]
    Note that when $\forall j, j' \in [N], k \in [K]$, $w_{j,k} = w_{j', k}$, we have $\forall j \in [N], k \in [K]$, $w_{j,k} = 1$ since we assume that $\forall k, \max_{j \in [N]} w_{j,k}=1$. As a result, this condition is a special case of the condition listed in Condition (2).
\end{proof}

\begin{proof} [Proof of Condition (3)]
    Suppose there exists an infinite improvement path. Then there must exist an improvement cycle
    \begin{equation*}
        \left(\bfpi(1), \bfpi(2), \dots, \bfpi(t), \bfpi(t+1) = \bfpi(1)\right).
    \end{equation*}
    Let $\calK = \{k(s): 1 \le s \le t\}$ be the set of all resources that appear in the improvement cycle.

    For each $1 \le s \le t$, since $\bfpi(s+1) = \left(k(s), \bfpi_{-j(s)}(s)\right)$ is an improvement step of $\bfpi(s)$ and $\mu_{\pi_{j(s)}(s)} = \mu_{k(s)}$ by condition, setting $j = j(s)$ for brevity, we have
    \begin{equation*}
        \frac{w_{j, \pi_{j}(s+1)}}{W_{\pi_{j}(s+1)}(\bfpi(s+1))}  = \frac{w_{j, k(s)}}{w_{j, k(s)} + W_{k(s)}\left(\bfpi(s)\right)} > \frac{w_{j, \pi_j(s)}}{W_{\pi_{j}(s)}(\bfpi(s))}.
    \end{equation*}
    As a result, subtract the above inequality by 1 in both sides, we can get that
    \begin{equation*}
        -\frac{W_{\pi_{j}(s+1)}\left(\bfpi(s)\right)}{W_{\pi_{j}(s+1)}(\bfpi(s+1))} = -\frac{W_{k(s)}\left(\bfpi(s)\right)}{w_{j, k(s)} + W_{k(s)}\left(\bfpi(s)\right)} > - \frac{W_{\pi_{j}(s)}(\bfpi(s)) - w_{j, \pi_j(s)}}{W_{\pi_{j}(s)}(\bfpi(s))} = - \frac{W_{\pi_{j}(s)}(\bfpi(s+1))}{W_{\pi_{j}(s)}(\bfpi(s))}.
    \end{equation*}
    Hence
    \begin{equation*} \label{eq:proof-proposition-pne-1}
        \frac{W_{\pi_{j}(s)}(\bfpi(s + 1))}{W_{\pi_{j}(s)}(\bfpi(s))} > \frac{W_{\pi_{j}(s+1)}\left(\bfpi(s)\right)}{W_{\pi_{j}(s+1)}(\bfpi(s+1))}.
    \end{equation*}
    As a result, we must have $W_{\pi_{j}(s)}(\bfpi(s+1))W_{\pi_{j}(s+1)}(\bfpi(s+1))/W_{\pi_{j}(s)}(\bfpi(s)) > 0$. As a result, Define
    \begin{equation*}
        c_0 = \frac{1}{2}\min_{1 \le s \le t} \frac{W_{\pi_{j}(s)}(\bfpi(s+1))W_{\pi_{j}(s+1)}(\bfpi(s+1))}{W_{\pi_{j}(s)}(\bfpi(s))} > 0
    \end{equation*}
    and a potential function $\Phi(\bfpi(s))$ for any $1 \le s \le t$ as
    \begin{equation*}
        \Phi(\bfpi(s)) = \prod_{k \in \calK} \left(\bbI[W_k(\bfpi(s)) = 0] c_0 + \bbI[W_k(\bfpi(s)) > 0]W_k(\bfpi(s))\right).
    \end{equation*}
    By definition, we have $\Phi(s) > 0$ for any $1 \le s \le t$. Consider the change of the potential function for any round $s$ and $s+1$. Since $W_k(\bfpi(s))$ differs from $W_k(\bfpi(s + 1))$ only when $k = \pi_{j(s)}(s)$ and $k = \pi_{j(s)}(s+1)$. Let $j = j(s)$. By \cref{eq:proof-proposition-pne-1}, we have $W_{\pi_{j}(s)}(\bfpi(s)), W_{\pi_{j}(s+1)}(\bfpi(s+1)), W_{\pi_{j}(s)}(\bfpi(s+1)) > 0$. Consider the two cases when $W_{\pi_{j}(s+1)}\left(\bfpi(s)\right) = 0$ or $W_{\pi_{j}(s+1)}\left(\bfpi(s)\right) > 0$.
    \begin{itemize}
        \item When $W_{\pi_{j}(s+1)}\left(\bfpi(s)\right) = 0$, we have
        \begin{equation*}
            \frac{\Phi(\bfpi(s + 1))}{\Phi(\bfpi(s))} = \frac{W_{\pi_{j}(s)}(\bfpi(s + 1))W_{\pi_{j}(s+1)}(\bfpi(s + 1))}{W_{\pi_{j}(s)}(\bfpi(s)) \cdot c_0} > 2
        \end{equation*}
        by the definition of $c_0$.
        \item When $W_{\pi_{j}(s+1)}\left(\bfpi(s)\right) > 0$, we have
        \begin{equation*}
            \frac{\Phi(\bfpi(s + 1))}{\Phi(\bfpi(s))} = \frac{W_{\pi_{j}(s)}(\bfpi(s + 1))W_{\pi_{j}(s+1)}(\bfpi(s + 1))}{W_{\pi_{j}(s)}(\bfpi(s)) W_{\pi_{j}(s+1)}\left(\bfpi(s)\right)} > 1
        \end{equation*}
        by \cref{eq:proof-proposition-pne-1}.
    \end{itemize}
    As a result, we have that $\Phi(\bfpi(s + 1)) > \Phi(\bfpi(s))$ for all $1 \le s \le t$. Therefore,
    \begin{equation*}
        \Phi(\bfpi(1)) = \Phi(\bfpi(t + 1)) > \Phi(\bfpi(t)) > \dots > \Phi(\bfpi(1)),
    \end{equation*}
    which leads to a contradiction. Therefore, there does not exist any infinite improvement path and PNE exists.
\end{proof}

\subsection{proof of \cref{thrm:poa}} \label{sect:proof-poa}
\begin{proof}
    In the PPA-Game, the strategy profile that maximizes social welfare assigns all players to the top-$N$ resources if $N \leq K$, and if $N > K$, the top $K$ players select all $K$ resources. Consequently, the maximal social welfare is $\sum_{i=1}^{\min\{N, K\}}\mu_i$. Denote $W^{\max} = \sum_{i=1}^{\min\{N,K\}} \mu_i$ as the numerator of PoA given in \cref{eq:poa}.
    
    Suppose $\bfpi$ is a PNE. Let $\calM = \{1, 2, \dots, \min\{N, K\}\}$ and $\calC = \{j \in [K]: \exists i \in [N], \pi_i = j\}$.
    As a result, by the definition of social welfare and the utility function given in \cref{eq:game-formulation}, we have that
    $$
    W^{\max} = \sum_{j \in \calM} \mu_j \quad \text{and} \quad W_{\calG}(\bfpi) = \sum_{j \in \calC}\mu_j = \sum_{i=1}^N \calU_i(\bfpi).
    $$
    Note that since $\sum_{i=1}^N \calU_i(\bfpi) = W_{\calG}(\bfpi)$. There must exist one player $i^*$ such that $\calU_i(\bfpi) \le W_{\calG}(\bfpi) / N$. Since $\bfpi$ is a PNE, player $i^*$ will not benefit by deviation, which means that $\forall j \in \calM \backslash \calC$, $\mu_j \le W_{\calG}(\bfpi) / N$. As a result,
    $$
    W^{\max} = \sum_{j \in \calM} \mu_j \le \sum_{j \in \calC}\mu_j + \sum_{j \in \calM \backslash \calC} \mu_j \le W_{\calG}(\bfpi) + |\calM \backslash \calC| \cdot W_{\calG}(\bfpi) / N.
    $$
    Note that we must have $1 \in \calM$ and $1 \in \calC$. Therefore, $|\calM\backslash\calC| \le |\calM| - 1 = \min\{N, K\} - 1$. This gives us
    $$
    \frac{W^{\max}}{W_{\calG}(\bfpi)} \le 1 + \frac{|\calM \backslash \calC|}{N} \le 1 + \frac{\min\{N, K\} - 1}{N}.
    $$
    Now the claim follows.
\end{proof}

\subsection{Proof of \cref{thrm:game-perturbed}} \label{sect:proof-game-perturbed}
\begin{lemma} \label{lemma:ne-property}
    For any Nash equilibrium $\bfpi$ of $\calG$, there exists an index $k$, such that (1) any arm with index larger than $k$ is not chosen by any player and (2) any arm with index not larger than $k$ is chosen by at least one player, \textit{i.e.},
    \begin{equation*} \label{eq:lemma-ne-property}
        \forall k' > k, \forall j \in [N], \pi_j \ne k' \quad \text{and} \quad \forall k' \le k, \exists j \in [N], \pi_j = k'.
    \end{equation*}
\end{lemma}
\begin{proof}
    Suppose the condition in \cref{eq:lemma-ne-property} does not hold for a Nash equilibrium $\bfpi$. Then there must exist two indices, $k_1 < k_2$ such that arm $k_1$ is not chosen by any player and arm $k_2$ is chosen by at least one player. Now any player that pulls arm $k_2$ can increase his reward by deviating to arm $k_1$, which leads to a contradiction since $\bfpi$ is a Nash equilibrium.
\end{proof}

Now we can prove \cref{thrm:game-perturbed}.
\begin{proof}[Proof of \cref{thrm:game-perturbed}]

    (a) We first demonstrate that under the conditions $\Delta < \delta/(4K)$ and $\Gamma \le \delta / (4N)$, all Nash equilibria of $\calG$ are also equilibria of $\calG'$ and vice versa.

    Note that for all $j \in N$ and $\bfpi \in [K]^N$, by setting $k=\pi_j$, we have
    \begin{equation*}
        \begin{aligned}
            \left|\calU_j(\bfpi) - \calU'_j(\bfpi)\right| & = \left|\frac{\left(\hat{\mu}_{j,k} - \mu_k\right) w_{j,k}}{\sum_{j'=1}^N \bbI[\pi_{j'}=k]w_{j',k}}+\gamma_{j,k}\right| \le \left|\frac{\left(\hat{\mu}_{j,k} - \mu_k\right) w_{j,k}}{\sum_{j'=1}^N \bbI[\pi_{j'}=k]w_{j',k}}\right|+\left|\gamma_{j,k}\right| \\
            & \le \left|\hat{\mu}_{j,k} - \mu_k\right| + |\gamma_{j,k}| \le \Delta + \Gamma \le \frac{\delta}{4} + \frac{\delta}{8} = \frac{3\delta}{8}.
        \end{aligned}
    \end{equation*}

    On the one hand, suppose $\bfpi$ is a Nash equilibrium of $\calG$. By the definition of $\delta$ and $\delta^{\PNEBound}$ in \cref{eq:delta}, we have that
    \begin{equation*}
        \forall j \in [N], \pi' \in [K] \text{ and } k \ne \pi_j, \quad \calU_j(\bfpi) - \calU_j(\pi', \bfpi_{-j}) \ge \delta^{\PNEBound} \ge \delta.
    \end{equation*}
    As a result, $\forall j \in [N], \pi' \in [K]$ and $k \ne \pi_j$, we have
    \begin{equation*}
        \begin{aligned}
            \calU'_j(\bfpi) - \calU'_j(\pi', \bfpi_{-j}) & = \calU_j(\bfpi) - \calU_j(\pi', \bfpi_{-j}) + \left(\calU'_j(\bfpi) - \calU_j(\bfpi)\right) + \left(\calU_j(\pi', \bfpi_{-j}) - \calU_j'(\pi', \bfpi_{-j})\right) \\
            & \ge \delta  - \left|\calU'_j(\bfpi) - \calU_j(\bfpi)\right| - \left|\calU_j(\pi', \bfpi_{-j}) - \calU_j'(\pi', \bfpi_{-j})\right| \ge \delta - \frac{3\delta}{8} - \frac{3\delta}{8} > 0.
        \end{aligned}
    \end{equation*}
    Therefore, $\bfpi$ is also a Nash equilibrium of $\calG'$.

    On the other hand, suppose $\bfpi$ is a Nash equilibrium of $\calG'$. If $\bfpi$ is not a Nash equilibrium of $\calG$, by \cref{eq:delta}, there must exist a strategy profile $\bfpi$, a player $j$ and a deviation $\pi' \ne \pi_j$ such that
    \begin{equation*}
        \calU_j(\pi', \bfpi_{-j}) \ge \calU_j(\bfpi) + \delta^{\NOPNEBound} \ge \calU_j(\bfpi) + \delta.
    \end{equation*}
    As a result, we have
    \begin{equation*}
        \begin{aligned}
             \calU'_j(\pi', \bfpi_{-j}) - \calU'_j(\bfpi) & = \calU_j(\pi', \bfpi_{-j}) - \calU_j(\bfpi) + \left(\calU_j(\bfpi) - \calU'_j(\bfpi)\right) + \left(\calU_j'(\pi', \bfpi_{-j}) - \calU_j(\pi', \bfpi_{-j})\right) \\
            & \ge \delta  - \left|\calU'_j(\bfpi) - \calU_j(\bfpi)\right| - \left|\calU_j(\pi', \bfpi_{-j}) - \calU_j'(\pi', \bfpi_{-j})\right| \ge \delta - \frac{3\delta}{8} - \frac{3\delta}{8} > 0,
        \end{aligned}
    \end{equation*}
    which leads to a contradiction since $\bfpi$ is a Nash equilibrium of $\calG$. Therefore, $\bfpi$ is also a Nash equilibrium of $\calG$.

    (b) We then demonstrate that any most efficient Nash equilibrium $\bfpi$ of $\calG'$ is also a most efficient equilibrium of $\calG$.

    Suppose $\bfpi$ is a most efficient equilibrium of $\calG'$ but not a most efficient equilibrium of $\calG$. By the part of $(a)$, we have that $\bfpi$ is also a Nash equilibrium of $\calG$. By \cref{lemma:ne-property}, we have that there exists $k$, such that 
    \begin{equation*}
        \forall k' > k, \forall j \in [N], \pi_j \ne k' \quad \text{and} \quad \forall k' \le k, \exists j \in [N], \pi_j = k'.
    \end{equation*}
    And we have the total welfare $W_{\calG'}(\bfpi)$ under the game $\calG'$ is given by the following bounds:
    \begin{equation*}
        \begin{aligned}
            W_{\calG'}(\bfpi) & = \sum_{j=1}^N \calU'_j(\bfpi) = \sum_{j=1}^N \left(\frac{\hat{\mu}_{j,\pi_j} w_{j,\pi_j}}{\sum_{j'=1}^N \bbI[\pi_{j'}=\pi_j]w_{j',\pi_j}}+\gamma_{j,\pi_j}\right) \\
            & \le \sum_{j=1}^N \frac{(\mu_{\pi_j} + \Delta) w_{j,\pi_j}}{\sum_{j'=1}^N \bbI[\pi_{j'}=\pi_j]w_{j',\pi_j}}+ N \Gamma \\
            & = \sum_{k'=1}^k \mu_{k'} + k\Delta + N \Gamma \le \sum_{k'=1}^k \mu_{k'} + K\Delta + N \Gamma
        \end{aligned}
    \end{equation*}
    and
    \begin{small}
        $$
        W_{\calG'}(\bfpi) = \sum_{j=1}^N \left(\frac{\hat{\mu}_{j,\pi_j} w_{j,\pi_j}}{\sum_{j'=1}^N \bbI[\pi_{j'}=\pi_j]w_{j',\pi_j}}+\gamma_{j,\pi_j}\right) \ge \sum_{j=1}^N \frac{(\mu_{\pi_j} - \Delta) w_{j,\pi_j}}{\sum_{j'=1}^N \bbI[\pi_{j'}=\pi_j]w_{j',\pi_j}} = \sum_{k'=1}^k \mu_{k'} - k\Delta \ge \sum_{k'=1}^k \mu_{k'} - K\Delta.
        $$
    \end{small}

    Because $\bfpi$ is not a most efficient equilibrium of $\calG$ by assumption, there must exist another equilibrium $\tilde{\bfpi}$ of $\calG$ such that $W_{\calG}(\tilde{\bfpi}) > W_{\calG}(\bfpi)$. By \cref{lemma:ne-property}, there must exist $\tilde{k} > k$ such that in the equilibrium $\tilde{\bfpi}$, all the arms in $[\tilde{k}]$ is chosen by at least one player and all other arms are not chosen by any player. By part (a), we have that $\tilde{\bfpi}$ is also an equilibrium of $\calG'$, and we have
    \begin{equation*}
        W_{\calG'}(\tilde{\bfpi}) - W_{\calG'}(\bfpi) \ge \sum_{k'=1}^{\tilde{k}} \mu_{k'} - K\Delta - \left(\sum_{k'=1}^k \mu_{k'} + K\Delta + N \Gamma\right) \ge \mu_{\tilde{k}} - 2K\Delta - N\Gamma.
    \end{equation*}
    Note that we must have $\mu_{\tilde{k}} \ge \delta$ because arm $\tilde{k}$ is chosen by at least one player $j$ in the equilibrium $\tilde{\bfpi}$ \textit{i.e.}, $\calG$ and player $j$'s any deviation will incur a loss not larger than $\mu_{\tilde{k}}$. As a result,
    \begin{equation*}
        W_{\calG'}(\tilde{\bfpi}) - W_{\calG'}(\bfpi) \ge \mu_{\tilde{k}} - 2K\Delta - N\Gamma \ge \delta - 2K \cdot \frac{\delta}{4K} - N \cdot \frac{\delta}{4N} = \frac{\delta}{4} > 0.
    \end{equation*}
    This leads to a contradiction since $\bfpi$ is a most efficient equilibrium of $\calG'$. Now the claim follows.
    
    (c) We finally demonstrate that the most efficient Nash equilibrium of $\calG'$ is unique with probability $1$.

    We have
    \begin{equation*} \label{eq:proof-unique}
        \begin{aligned}
            & \, \bbP\left[\exists \bfpi', \bfpi'' \in [K]^N, \bfpi' \text{ and } \bfpi'' \text{ are the most efficient Nash equilibria of } \calG', W_{\calG'}(\bfpi') = W_{\calG'}(\bfpi'')\right] \\
            \le & \, \sum_{\bfpi', \bfpi'' \in [K]^N} \bbP[W_{\calG'}(\bfpi') = W_{\calG'}(\bfpi'')].
        \end{aligned}
    \end{equation*}
    Note that $W_{\calG'}(\bfpi') = W_{\calG'}(\bfpi'')$ if and only if
    \begin{equation*}
        \sum_{j=1}^N \left(\gamma_{j,\pi'_j} - \gamma_{j, \pi''_j}\right) = \sum_{j=1}^N\left(\frac{\hat{\mu}_{j,\pi'_j} w_{j,\pi'_j}}{\sum_{j'=1}^N \bbI[\pi'_{j'}=\pi'_j]w_{j',\pi'_j}} - \frac{\hat{\mu}_{j,\pi''_j} w_{j,\pi''_j}}{\sum_{j'=1}^N \bbI[\pi''_{j'}=\pi''_j]w_{j',\pi''_j}}\right).
    \end{equation*}
    Note that this occurs with probability $0$ because $\gamma_{j,k}$ are sampled uniformly from $[0, \Gamma]$ and the right-hand side of the above equation is a constant when $\bfpi'$ and $\bfpi''$ are fixed. As a result, we have $\forall \bfpi', \bfpi'' \in [K]^N$, $\bbP[W_{\calG'}(\bfpi') = W_{\calG'}(\bfpi'')] = 0$ and the left-hand side of \cref{eq:proof-unique} is $0$. Now the claim follows.
\end{proof}

\subsection{Proof of \cref{thrm:regret}} \label{sect:proof-regret}
\paragraph{Further Notations}
Let $\calZ$ be the space of all possible states of $\calM(\epsilon; \calG')$. In addition, define $\calD$ as the set of all states in the Markov process where all players are discontent.
\begin{equation*} \label{eq:calD}
    \calD = \left\{Z = \{(m_j, a_j, u_j)\}_{j=1}^N\in \calZ: \forall j \in [N], m_j = \discontent\right\}.
\end{equation*}
We first note that similar to \citep{pradelski2012learning}, the transition probability from any state $Z_1 \in \calZ \backslash \calD$ to $Z_2 \in \calD$ and the transition probability from any state $Z_2 \in \calD$ to $Z_1 \in \calZ \backslash \calD$ are independent of the $Z_2$. As a result, we can view all states in $\calD$ as a single state. The transition probability from $\calD$ to another state $Z \in \calZ \backslash \calD$ is the same as the probability from any state $Z_1 \in \calD$ to $Z$. We slightly abuse the notation here to denote the Markov process as $\calM(\epsilon; \calG')$ after combining all the states in $\calD$ as a single state.

For any $\epsilon > 0$, since the Markov process $\calM(\epsilon; \calG')$ is ergodic, the stationary distribution is unique. We denote the unique stationary distribution of $\calM(\epsilon; \calG')$ as $p(\epsilon; \calG')$ and $p_Z(\epsilon; \calG')$ is the probability of state $Z \in \calZ$ in the stationary distribution. Let $A(\epsilon; \calG')$ be the $1/8$-mixing time \citep{chung2012chernoff} of the Markov process $\calM(\epsilon; \calG')$ and $A(\epsilon; \calG') = \min \{t: \max_{x}\|x M^t - p\|_{\text{TV}} \le 1/8\}$ where $x$ is an arbitrary initial distribution on $\calZ$ and $\|\cdot\|_{\text{TV}}$ denotes the total variation distance.

We first note that the process $\calM(\epsilon; \calG')$ is a regular perturbed Markov process defined as follows.
\begin{definition} [Regular Perturbed Markov Process \citep{young1993evolution}] \label{defn:RP-MP}
    A Markov process $\calM_\epsilon$ is called a regular perturbed Markov process if $\calM_\epsilon$ is ergodic for all $\epsilon > 0$ and for every $Z, Z' \in \calZ$ we have
    \begin{equation*}
        \lim_{\epsilon \rightarrow 0^+} \calM_{\epsilon}(Z \rightarrow Z') = \calM_{0}(Z \rightarrow Z')
    \end{equation*}
    and if $\calM_\epsilon(Z \rightarrow Z') > 0$ for some $\epsilon > 0$ then
    \begin{equation*}
        0 < \lim_{\epsilon \rightarrow 0^+} \frac{\calM_\epsilon(Z \rightarrow Z')}{\epsilon^{r(Z \rightarrow Z')}} < \infty
    \end{equation*}
    for some real non-negative $r(Z \rightarrow Z')$ called the resistance of the transition $Z \rightarrow Z'$.
\end{definition}
For such regular perturbed Markov process, the stochastic stability of a state is defined as follows.
\begin{definition} [Stochastic Stability \citep{young1993evolution}] \label{defn:stochastic-stable}
    A state $Z \in \calZ$ is stochastically stable relative to the process $\calM_\epsilon$ if
    \begin{equation*}
        \lim_{\epsilon \rightarrow 0^+} p_Z(\epsilon) > 0.
    \end{equation*}
    Here $p_Z(\epsilon)$ denotes the probability of the state $Z$ in the stationary distribution of $\calM_\epsilon$.
\end{definition}

In addition, we need the definition of interdependent game.

\begin{definition} [Interdependent Game \citep{pradelski2012learning}]
    The game $\calG$ is interdependent if for any strategy profile $\bfpi$ and every proper subset of players $\emptyset \subsetneq J \subsetneq [N]$, there exists some player $j \not\in J$ and a choice of strategies $\bfpi'_{J}$ such that $\calU_j(\bfpi'_J, \bfpi_{-J}) \ne \calU_j(\bfpi)$. Here the strategy profile $(\bfpi'_J, \bfpi_{-J})$ denotes the profile where all players in the set $J$ deviate to $\bfpi'_J$ and other players' strategies remain unchanged.
\end{definition}

\paragraph{Important Lemmas}
With these notations, we further need two lemmas.

\begin{lemma} \label{lemma:stationary-distribution}
    Suppose the conditions in \cref{thrm:regret} hold. Let $\bfpi^*$ be the unique most efficient Nash equilibrium \textit{w.r.t.} the game $\calG'$ according to \cref{thrm:game-perturbed}. Then with probability $1$ the following holds. The state
    \begin{equation*} \label{eq:z-star}
        Z^* = \{(m_j^*, a_j^*, u_j^*)\}_{j \in [N]} \text{ with } \forall j \in [N], m_j^* = \content, a_j^* = \pi_j^*, u_j^* = \calU'_j(\bfpi^*)
    \end{equation*}
    is the unique stochastically stable state of the regular perturbed Markov process $\calM(\epsilon; \calG')$. In addition, for any $0 < \alpha < 1 / 2$, there exists a small enough $\epsilon$ such that the probability of $Z^*$ in the stationary distribution $p_{Z^*}(\epsilon; \calG')$ is larger than $1 / (2(1-\alpha))$.
\end{lemma}
\begin{proof}
    (a) Firstly, we demonstrate that note that $\calG'$ is an interdependent game based on the assumption that $w_{j, k} > 0$ for all $j \in [N]$ and $k \in [K]$.
    
    For any strategy profile $\bfpi$ and every proper subset of players $\emptyset \subsetneq J \subsetneq [N]$, consider any player $j \not\in J$ and any player $j' \in \calJ$. If $\pi_j = \pi_{j'}$, then the payoff of player $j'$ will change when player $j$ chooses another strategy instead of $\pi_j$ due to $w_{j, \pi_j} > 0$. Similarly, if $\pi_j \ne \pi_{j'}$, then the payoff of player $j'$ will also change when player $j$ chooses the strategy $\pi_{j'}$ due to $w_{j, \pi_{j'}} > 0$. As a result, $\calG'$ is an independent game by definition.

    (b) We further demonstrate the functions $F(\cdot)$ and $G(\cdot)$ in \cref{eq:F-and-G} meet certain conditions. 
    
    For any $j \in [N]$, define the set $V_j = \left\{\calU'_{j}(\bfpi): \bfpi \in [K]^N\right\}$ where $\calU'_j(\bfpi)$ is provided in \cref{eq:game-formulation-perturbed}. As a result,
    \begin{equation*}
        \forall u' \in V_j,  u' \le \max_{k}(\hat{\mu}_{j,k} + \gamma_{j,k}) \le 1 + \Gamma, \quad \text{and} \quad \forall u', u'' \in V_j, \left|u' - u''\right| \le \max_{u' \in V_j}u' \le 1 + \Gamma.
    \end{equation*}
    Therefore, $\forall u' \in V_j, \quad$,
    \begin{equation*}
         F(u') = -\frac{u'}{(4+3\Gamma)N}+\frac{1}{3N} > -\frac{1 + \Gamma}{(4 + 3\Gamma)N} + \frac{1}{3N} > 0 \text{ and } F(u') = -\frac{u'}{(4+3\Gamma)N}+\frac{1}{3N} < \frac{1}{3N} < \frac{1}{2N}.
    \end{equation*}
    In addition, $\forall u', u'' \in V_j, u' \ge u''$,
    \begin{equation*}
        G(u' - u'') = -\frac{u' - u''}{4+3\Gamma} + \frac{1}{3} \ge -\frac{1 + \Gamma}{4+3\Gamma} + \frac{1}{3} > 0 \text{ and } G(u' - u'') = -\frac{u' - u''}{4+3\Gamma} + \frac{1}{3} < \frac{1}{3} < \frac{1}{2}.
    \end{equation*}
    
    (c) Finally, we could prove this lemma according to Theorem 1 in \citep{pradelski2012learning}. Note that our Markov process $\calM(\epsilon; \calG')$ follow exactly the one proposed in \citep{pradelski2012learning} and the only modifications are to ensure that the payoff calculation for each player follows the game $\calG'$. By part (a) and (b) of our proof, the conditions of Theorem 1 in \citep{pradelski2012learning} hold. As a result, since $\calG'$ have at least one PNE and the most efficient PNE is unique with probability $1$ according to \cref{thrm:game-perturbed}, the only stochastically stable state is $Z^*$ in \cref{eq:z-star}. As a result, when $\epsilon \rightarrow 0$, $p_{Z^*}(\epsilon; \calG') \rightarrow 1$. Now the claim follows.
\end{proof}

\begin{lemma} \label{lemma:}
    Suppose the conditions in \cref{thrm:regret} hold. Let $w_s$ be the set of rounds in the learning phase with length $c_2s^{\eta}$ and $Z^*$ be the unique stochastically stable state of $\calM(\epsilon; \calG')$ given by \cref{eq:z-star}. Let $Z(t)$ be status of the Markov process $\calM(\epsilon; \calG')$ at round $t$. Define the counter $H_s$ and the event $E_{2s}$ as
    \begin{equation*} \label{eq:event-2s}
        H_s = \sum_{t \in w_s} \bbI[Z(t) = Z^*] \quad \text{and} \quad E_{2s} = \left\{\sum_{i=1}^s H_i > \frac{L_s}{2}\right\}, L_s = \sum_{i = 1}^sc_2i^{\eta}.
    \end{equation*}
    For any $0 < \alpha < 1$, let $\epsilon$ be small enough that $p_{Z^*}(\epsilon; \calG') > 1 / (2(1-\alpha))$ according to \cref{lemma:stationary-distribution}. Then there exists a constant $C(\epsilon; \calG')$, such that
    \begin{equation*}
        \bbP[E_{2s}] \ge 1 - \left(\left(1 + \frac{C(\epsilon; \calG')}{\sqrt{p_{\calD}(\epsilon; \calG')}}\right)\exp\left(-\frac{\alpha^2c_2}{144A(\epsilon; \calG')}\left(p_{Z^*}(\epsilon; \calG') - \frac{1}{2(1-\alpha)}\right)\left(\frac{s}{2}\right)^{\eta}\right)\right)^s,
    \end{equation*}
    where $A(\epsilon; \calG')$ is the $1/8$-mixing time of $\calM(\epsilon; \calG')$.
\end{lemma}
\begin{proof}
    The proof follows from Lemma 6 of \citep{bistritz2021game}. Define $f: \calZ \rightarrow [0,1]$ as $f(Z^*) = 1$ and $f(Z) = 0, \forall Z \ne Z^*$. As a result, let $\nu = \bbE_{Z\sim p(\epsilon; \calG')}[f(Z)] = p_{Z^*}(\epsilon; \calG')$. The initial distribution $\phi$ is chosen that $\phi_{\calD} = 1$ and $\phi_{Z} = 0, \forall Z \ne \calD$. According to \cref{lemma:markov-process}, we have
    \begin{equation*} \label{eq:proof-hs}
        \bbP\left[H_s \le (1 - \alpha)p_{Z^*}(\epsilon; \calG')c_2s^{\eta}\right] \le \frac{C(\epsilon; \calG')}{\sqrt{p_{\calD}(\epsilon; \calG')}} \exp\left(-\frac{\alpha^2p_{Z^*}(\epsilon; \calG')c_2s^{\eta}}{72A(\epsilon; \calG')}\right),
    \end{equation*}
    Note that $\{H_s\}$ are independent. As a result, according to the Chernoff bound, we have $\forall \xi > 0$,
    \begin{equation*}
        \bbP\left[\bar{E}_{2s}\right] = \bbP\left[\sum_{i=1}^s H_i \le \frac{L_s}{2}\right] \le \exp\left(\frac{\xi L_s}{2}\right)\bbE\left[\prod_{i=1}^s\exp\left(-\xi H_{i}\right)\right] =  \exp\left(\frac{\xi L_s}{2}\right)\prod_{i=1}^s\bbE\left[\exp\left(-\xi H_{i}\right)\right].
    \end{equation*}
    For every $i$ such that $1 \le i \le s$, by \cref{eq:proof-hs} we have
    \begin{equation*}
        \begin{aligned}
            & \, \bbE[\exp(-\xi H_{i})] \\
            \le & \, \bbP\left[H_i \le (1 - \alpha)p_{Z^*}(\epsilon; \calG')c_2i^{\eta}\right]\bbE[\exp(-\xi H_{i}) \mid H_i \le (1 - \alpha)p_{Z^*}(\epsilon; \calG')c_2i^{\eta}] \\
            & \, + \bbP\left[H_i > (1 - \alpha)p_{Z^*}(\epsilon; \calG')c_2i^{\eta}\right]\bbE[\exp(-\xi H_{i}) \mid H_i > (1 - \alpha)p_{Z^*}(\epsilon; \calG')c_2i^{\eta}] \\
            \le & \, \frac{C(\epsilon; \calG')}{\sqrt{p_{\calD}(\epsilon; \calG')}} \exp\left(-\frac{\alpha^2p_{Z^*}(\epsilon; \calG')c_2i^{\eta}}{72A(\epsilon; \calG')}\right) + \exp\left(-\xi (1 - \alpha)p_{Z^*}(\epsilon; \calG')c_2i^{\eta}\right).
        \end{aligned}
    \end{equation*}
    By choosing $\xi = \alpha^2/\left(72(1-\alpha)A(\epsilon; \calG')\right)$, we have
    \begin{equation*}
        \bbE[\exp(-\xi H_{i})] \le \left(1 + \frac{C(\epsilon; \calG')}{\sqrt{p_{\calD}(\epsilon; \calG')}}\right)\exp\left(-\frac{\alpha^2p_{Z^*}(\epsilon; \calG')c_2i^{\eta}}{72A(\epsilon; \calG')}\right).
    \end{equation*}
    Therefore, we have
    \begin{equation*}
        \begin{aligned}
            \bbP[\bar{E}_{2s}] & \le \exp\left(\frac{\alpha^2 L_s}{144(1-\alpha)A(\epsilon; \calG')}\right) \prod_{i=1}^s\left(1 + \frac{C(\epsilon; \calG')}{\sqrt{p_{\calD}(\epsilon; \calG')}}\right)\exp\left(-\frac{\alpha^2p_{Z^*}(\epsilon; \calG')c_2i^{\eta}}{72A(\epsilon; \calG')}\right) \\
            & \le \exp\left(\frac{\alpha^2 L_s}{144(1-\alpha)A(\epsilon; \calG')}\right) \left(1 + \frac{C(\epsilon; \calG')}{\sqrt{p_{\calD}(\epsilon; \calG')}}\right)^s\exp\left(-\frac{\alpha^2p_{Z^*}(\epsilon; \calG')L_s}{72A(\epsilon; \calG')}\right) \\
            & \le \left(1 + \frac{C(\epsilon; \calG')}{\sqrt{p_{\calD}(\epsilon; \calG')}}\right)^s\exp\left(-\frac{\alpha^2L_s}{72A(\epsilon; \calG')}\left(p_{Z^*}(\epsilon; \calG') - \frac{1}{2(1-\alpha)}\right)\right).
        \end{aligned}
    \end{equation*}
    Since $\forall \eta > 0$,
    \begin{equation*}
        L_s = \sum_{i=1}^s c_2 i^\eta \ge c_2 \int_{0}^{s-1}i^\eta \mathrm{d} i = \frac{c_2(s-1)^{\eta+1}}{\eta+1} \ge c_2 \left(\frac{s}{2}\right)^{\eta + 1}
    \end{equation*}
    and $p_{Z^*}(\epsilon; \calG') \ge 1 / (2(1-\alpha))$, we have that
    \begin{equation*}
        \bbP[\bar{E}_{2s}]  \le \left(\left(1 + \frac{C(\epsilon; \calG')}{\sqrt{p_{\calD}(\epsilon; \calG')}}\right)\exp\left(-\frac{\alpha^2c_2}{144A(\epsilon; \calG')}\left(p_{Z^*}(\epsilon; \calG') - \frac{1}{2(1-\alpha)}\right)\left(\frac{s}{2}\right)^{\eta}\right)\right)^s.
    \end{equation*}
    Now the claim follows.
\end{proof}

Now we can prove \cref{thrm:regret}.

\begin{proof} [Proof of \cref{thrm:regret}]
    (a) For the exploration phase, define the clean event as
    \begin{equation*} \label{eq:event-1}
        E_1 = \left\{\forall j \in [N], k \in [K], \left|\hat{\mu}_{j,k} - \mu_k\right| < \frac{\delta}{4K}\right\}.
    \end{equation*}
    Since $c_1 \ge 8K^2\log T/\delta^2$, according to Hoeffding's inequality, we have that $\forall j \in [N], k \in [K]$,
    \begin{equation*}
        \bbP\left[\left|\hat{\mu}_{j, k} - \mu\right| > \frac{\delta}{4K} \right] \le 2\exp\left(-\frac{c_1\delta^2}{8K^2}\right) \le \frac{2}{T}.
    \end{equation*}
    As a result, we have
    \begin{equation*}
        \bbP(E_1) \ge 1 - \sum_{j=1}^N \sum_{k=1}^K  \bbP\left[\left|\hat{\mu}_{j, k} - \mu\right| > \frac{\delta}{4K} \right] \ge 1 - \frac{2NK}{T}.
    \end{equation*}
    The expected regret cost by the bad event $\bar{E}_1$ is bounded by $(2NK/T) \cdot T \mu_{\max} = 2NK \mu_{\max} = O(1)$. As a result, we focus on the clean event for the rest of the proof.

    (b) Consider the learning and exploitation phase. Under the clean event $E_1$, according to \cref{thrm:game-perturbed}, we have that the most efficient Nash equilibrium $\bfpi^*$ of game $\calG'$ is also a most efficient equilibrium of game $\calG$. The solution is also unique with probability $1$. As a result, according to \cref{lemma:stationary-distribution}, for any $0 < \alpha < 1$ there exists a small enough $\epsilon$ such that the state $Z^*$ in \cref{eq:z-star} is the unique stationary distribution of $\calM(\epsilon; \calG')$ and $p_{Z^*}(\epsilon; \calG') > 1 / (2(1 - \alpha))$.

    Suppose we have switched between the learning and exploitation phase $S$ times and the regret caused by the $s$-th phase is denoted as $\regret_{j,s}^{\learn}(T)$.
    As a result, the expected regret in the learning and exploitation phase is given by
    \begin{small}
    $$
    \begin{aligned}
        & \, \bbE\left[\regret_j^{\learn}(T) \mid E_1\right] \\
        \le & \, \sum_{s=1}^S \bbE\left[\regret_{j,s}^{\learn}(T) \mid E_1\right] = \sum_{s=1}^S \left(\bbE\left[\regret_{j,s}^{\learn}(T) \mid E_1, E_{2s}\right]\bbP(E_{2s}) + \bbE\left[\regret_{j,s}^{\learn}(T) \mid E_1, \bar{E}_{2s}\right](1 - \bbP(E_{2s}))\right) \\
        \le & \, \sum_{s=1}^S \left(\bbE\left[\regret_{j,s}^{\learn}(T) \mid E_1, E_{2s}\right] + \mu_{\max}\left(c_2s^\eta + c_3 2^s\right)\left(B(\epsilon; \calG')\exp\left(-\frac{\alpha^2c_2}{144A(\epsilon; \calG')}\left(p_{Z^*}(\epsilon; \calG') - \frac{1}{2(1-\alpha)}\right)\left(\frac{s}{2}\right)^{\eta}\right)\right)^s\right) \\
        \le & \, \mu_{\max}\left(\sum_{s=1}^S c_2 s^{\eta} + \sum_{s=1}^S\left(c_2s^\eta + c_3 2^s\right)\left(B(\epsilon; \calG')\exp\left(-\frac{\alpha^2c_2}{144A(\epsilon; \calG')}\left(p_{Z^*}(\epsilon; \calG') - \frac{1}{2(1-\alpha)}\right)\left(\frac{s}{2}\right)^{\eta}\right)\right)^s\right)
    \end{aligned}
    $$
    \end{small}
    where $B(\epsilon; \calG') = 1 + C(\epsilon; \calG')/\sqrt{p_{\calD}(\epsilon; \calG')}$. As a result, there exists a large enough $S_0 > 0$ such that
    \begin{equation*}
        \forall s > S_0, \quad B(\epsilon; \calG')\exp\left(-\frac{\alpha^2c_2}{144A(\epsilon; \calG')}\left(p_{Z^*}(\epsilon; \calG') - \frac{1}{2(1-\alpha)}\right)\left(\frac{s}{2}\right)^{\eta}\right) \le \frac{1}{2}.
    \end{equation*}
    As a result, when $T$ is large enough (dependent on $\epsilon$ and $\calG'$), we have
    \begin{small}
    \begin{equation*}
        \begin{aligned}
            & \, \bbE\left[\regret_j^{\learn}(T) \mid E_1\right] \\
            \le & \, \mu_{\max}\left(\sum_{s=1}^S c_2 s^{\eta} + \sum_{s=1}^{S_0}\left(c_2s^\eta + c_3 2^s\right)\left(B(\epsilon; \calG')\exp\left(-\frac{\alpha^2c_2}{144A(\epsilon; \calG')}\left(p_{Z^*}(\epsilon; \calG') - \frac{1}{2(1-\alpha)}\right)\left(\frac{s}{2}\right)^{\eta}\right)\right)^s\right) \\
            & \, + \mu_{\max}\left(\sum_{s=S_0+1}^S\left(c_2s^\eta + c_3 2^s\right) \left(\frac{1}{2}\right)^s\right) \\
            \le & \, 2\mu_{\max}\left(\sum_{s=1}^S c_2 s^{\eta} + c_3S\right) + O(1) \\
            \le & \, 2\mu_{\max} \left(\int_1^{S+1}c_2s^{\eta}\mathrm{d} s + c_3S\right) + O(1) \\
            \le & \, 2\mu_{\max}\left(\frac{c_2(S+1)^{\eta + 1}}{\eta+1}+c_3S\right) + O(1).
        \end{aligned}
    \end{equation*}
    \end{small}
    By the definition of $S$, we have $T \le \sum_{s=1}^S c_32^s = c_3(2^{S+1} - 2)$ and $S = O(\log T)$. As a result, we have
    \begin{equation*}
        \bbE\left[\regret_j^{\learn}(T) \mid E_1\right] \le 2\mu_{\max}\left(\frac{c_2(S+1)^{\eta + 1}}{\eta+1}+c_3S\right) \le O\left(\mu_{\max}\left(c_2\log^{1 + \eta} T + c_3\log T\right)\right).
    \end{equation*}
    
    (c) Combining the above two parts, when $T$ is large, the expected regret is given by
    \begin{equation*}
        \begin{aligned}
            \bbE[\regret_j(T)] & = \bbE\left[\regret_j^{\init}(T)\right] + \bbE\left[\regret_j^{\learn}(T)\right] \\
            & \le \mu_{\max}c_1 + \bbE\left[\regret_j^{\learn}(T) \mid E_1\right]\bbP(E_1) + \bbE\left[\regret_j^{\learn}(T) \mid \bar{E}_1 \right](1 - \bbP(E_1)) \\
            & \le \mu_{\max}c_1 + \bbE\left[\regret_j^{\learn}(T) \mid E_1\right] + 2NK\mu_{\max} \\
            & \le \mu_{\max}c_1 + O\left(\mu_{\max}\left(c_2\log^{1 + \eta} T + c_3\log T\right)\right) \\
            & = O\left(\mu_{\max}\left(c_1 + c_2\log^{1 + \eta} T + c_3\log T\right)\right).
        \end{aligned}
    \end{equation*}
    Similar to the above proof, we have that
    \begin{equation*}
        \begin{aligned}
            \bbE[\noneq(T)] & = \bbE\left[\noneq^{\init}(T)\right] + \bbE\left[\noneq^{\learn}(T)\right] \\
            & \le c_1 + \bbE\left[\noneq^{\learn}(T) \mid E_1\right]\bbP(E_1) + \bbE\left[\noneq^{\learn}(T) \mid \bar{E}_1 \right](1 - \bbP(E_1)) \\
            & \le c_1 + \bbE\left[\noneq^{\learn}(T) \mid E_1\right] + 2NK \\
            & \le c_1 + O\left(c_2\log^{1 + \eta} T + c_3\log T\right) \\
            & = O\left(c_1 + c_2\log^{1 + \eta} T + c_3\log T\right).
        \end{aligned}
    \end{equation*}
    Here $\noneq^{\init}(T)$ and $\noneq^{\learn}(T)$ represents the number of rounds where players do not follow a most efficient PNE in the exploration phase and learning and exploitation phase, respectively. The bounds on $\noneq^{\init}(T)$ and $\noneq^{\learn}(T)$ share the same techniques with $\regret_j^{\init}(T)$ and $\regret_j^{\learn}(T)$. Now the claim follows.
\end{proof}

\subsection{Important Lemmas}
\begin{lemma} [Theorem 3.1 in \citep{chung2012chernoff}] \label{lemma:markov-process}
    Let $M$ be an ergodic Markov chain with state space $\calZ$ and stationary distribution $p$. Let $A$ be its $1/8$-mixing time. Let $(Z_1, \dots, Z_t)$ denote a $t$-step random walk on $M$ starting from an initial distribution $\phi$ on $\calZ$, \textit{i.e.}, $Z_1 \sim \phi$. For every $i \in [t]$, let $f_i: \calZ \rightarrow [0, 1]$ be a weight function at step $i$ such that the expected weight $\bbE_{Z \sim p}[f_i(Z)] = \nu$ for all $i$. Define the total weight of the walk $(Z_1, \dots, Z_t)$ by $Y = \sum_{i=1}^t f_i(Z_i)$. There exists some constant $c$ (which is independent of $\nu$ and $\alpha$) such that $\forall 0 < \alpha < 1$,
    \begin{equation*}
        \bbP[Y \le (1 - \alpha)\nu t] \le c \|\phi\|_p\exp\left(-\alpha^2\nu t/(72A)\right).
    \end{equation*}
    Here $A = \min \{t: \max_{x}\|x M^t - p\|_{\text{TV}} \le 1/8\}$ where $x$ is an arbitrary initial distribution on $\calZ$ and $\|\cdot\|_{\text{TV}}$ denotes the total variation distance. In addition, $\|\phi\|_p = \sqrt{\sum_{Z \in \calZ}\phi^2(Z)/p(Z)}$.
\end{lemma}

\end{document}